\documentclass[lettersize,journal]{IEEEtran}
\usepackage{amsmath,amsfonts}
\usepackage{algorithm}
\usepackage{array}
\usepackage[caption=false,font=normalsize,labelfont=sf,textfont=sf]{subfig}
\usepackage{textcomp}
\usepackage{stfloats}
\usepackage{url}
\usepackage{verbatim}
\usepackage{graphicx}
\usepackage{cite}
\usepackage{amssymb,amsthm}
\usepackage{stfloats,color}
\usepackage{wasysym,epsfig}
\usepackage{enumerate,algpseudocode,balance,bm,nicefrac}
\usepackage{booktabs} 
\usepackage{multirow}
\usepackage{epstopdf}
\allowdisplaybreaks

\newtheorem{lemma}{Lemma}
\newtheorem{thm}{\bf Theorem}
\newtheorem{assumption}{\bf Assumption}

\theoremstyle{remark}
\theoremstyle{remark}\newtheorem{remark}{Remark}

\begin{document}


\title{\Large \bf Self-triggered Consensus Control of 
	Multi-agent Systems from Data}

\author{Yifei Li, Xin Wang, Jian Sun,~\IEEEmembership{Senior Member,~IEEE},\\ Gang Wang,~\IEEEmembership{Senior Member,~IEEE},
	and
	Jie Chen,~\IEEEmembership{Fellow,~IEEE}
\thanks{The work was partially supported by the National Key R\&D Program of China under Grant 2021YFB1714800, and by the National Natural Science Foundation of China under Grants 62173034, 61925303, 62088101, U20B2073.
}
	\thanks{Yifei Li, Xin Wang, Jian Sun, and Gang Wang are with the National Key Lab of Autonomous Intelligent Unmanned Systems, Beijing Institute of Technology, Beijing 100081, China, and with the Beijing Institute of Technology Chongqing Innovation Center, Chongqing 401120, China (e-mail: liyifei@bit.edu.cn; xinwang@bit.edu.cn; sunjian@bit.edu.cn; gangwang@bit.edu.cn).
		
	
	Jie Chen is with the Department of Control Science and Engineering, Tongji University, Shanghai 201804, China, and with  the National Key Lab of Autonomous Intelligent Unmanned Systems, Beijing Institute of Technology, Beijing 100081, China 	
	(e-mail: chenjie@bit.edu.cn).
}
}



\maketitle

\begin{abstract}
This paper considers self-triggered consensus control of unknown linear multi-agent systems (MASs). Self-triggering mechanisms (STMs) are widely used in MASs, thanks to their advantages in avoiding continuous monitoring and saving computing and communication resources. However, existing results require the knowledge of system matrices, which are difficult to obtain in real-world settings. To address this challenge,
we present a data-driven approach to designing STMs for unknown MASs {building upon the model-based solutions}.
Our approach leverages a system lifting method, which allows us to derive a data-driven representation for the MAS. Subsequently, a data-driven self-triggered consensus control (STC) scheme is designed, which combines a data-driven STM with a state feedback control law. We establish a data-based stability criterion for asymptotic consensus of the closed-loop MAS
in terms of linear matrix inequalities, whose solution provides a matrix for the STM as well as a stabilizing controller gain. In the presence of external disturbances, a model-based STC scheme is put forth for $\mathcal{H}_{\infty}$-consensus of MASs, serving as a baseline for the data-driven STC. Numerical tests are conducted to validate the correctness of the data- and model-based STC approaches.
Our data-driven approach demonstrates a superior trade-off between control performance and communication efficiency from finite, noisy data relative to the system identification-based one. 
\end{abstract}

\allowdisplaybreaks


\section{Introduction}
Multi-agent systems (MASs), such as multi-robots, vehicle platoons, and sensor networks, have received remarkable attention from both academia and industry, because of their potential to perform complex tasks collaboratively  \cite{Bi2022,Arslan2019,Ganguli2009,Francescelli2013}.
Consensus is a crucial issue of MASs, which can be classified into leaderless consensus \cite{Olfati2004,Mei2016}, reaching the common state of all agents or leader-following consensus \cite{Li2013,Liu2021}, tracking the state of a leader, which can be a real agent or a virtual one. This paper contributes to the leader-following consensus, aiming at designing distributed control algorithms given only local information of leader-following MASs.

To cope with limited communication bandwidth and energy of agents, an effective energy- and bandwidth-saving approach, termed self-triggered control, has been proposed for networked control systems, see e.g., \cite{Velasco2003,Mazo2009}.
The underlying idea is to design a controller along with a self-triggered transmission policy, a.k.a., self-triggering mechanism (STM), to improve the efficiency of resource utilization and prolong the life-cycle of sensors by sampling and transmitting only when necessary.
Using STM, the next triggering instant is predetermined based on the available information at the current triggering instant.
Since transmission times are scheduled in advance, sensors can be switched to sleep mode during inter-event triggering intervals, thus continuous
 monitoring can be avoided to save energy and computing resources \cite{Heemels2012}.
Recently, {self-triggered control has been widely studied for single-integrator MASs in \cite{Yi2019}, general MASs in \cite{Cui2023},} leader-following consensus under DoS attacks in \cite{Zhang2022self} or with byzantine adversaries in \cite{Zegers2022}, and distributed model predictive control (MPC) in \cite{Mi2020}. 


Note that the above-mentioned self-triggered control studies are all model-based, in the sense that both controllers and STMs are designed based on explicit system models.
Indeed in real-world applications, it is challenging to obtain an accurate system model.
The emerging data-driven control, on the other hand, provides an end-to-end paradigm for direct synthesis and analysis of unknown sampled-data control systems using purely data \cite{Markovsky2021}, without resorting to any intermediate procedure, e.g., identifying a system model \cite{Ljung1986}.
Notably, a cornerstone of data-driven control is the fundamental lemma \cite{Willems2005}, which rigorously characterizes conditions for describing the behavior of a linear time-invariant system by a single trajectory. 
{
While data-driven control and analysis draw insights from model-based derivations, its primary emphasis is on replacing process models with data. 
Up to date, this control paradigm has been extended in both theory as well as application areas,}
 including robust control in \cite{Persis2020,Waarde2022,Liu2022dos}, aperiodic sampled-data control in \cite{Wildhagen2022,Wang2021}, MPC in \cite{Coulson2019,Berberich2021}, quantized control in \cite{Zhao2022data}, and event-triggered control in \cite{Persis2022,Qi2022,Wang2023jas}, to name a few. 
Self-triggered control from data, self-triggered consensus control (STC) in particular, remains relatively less explored. 
{Recently, the data-driven self-triggered control problem of linear systems has been studied in \cite{Wenjie2023,Wang2022self} via trajectory prediction and a behavioral approach, respectively.}

All the aforementioned results on data-driven control have mainly focused on single systems. 
The widespread applications of MASs call for research efforts in distributed data-driven  solutions, that is, 
developing distributed controllers directly from local data.
Indeed, several distributed data-driven control strategies have been developed in the literature for a range of control problems.
Distributed data-driven optimal control is investigated in \cite{Baggio2021}. Assuming known noise, distributed data-driven controllers are developed to achieve output synchronization in \cite{Jiao2021} for heterogeneous MASs. 
Distributed iterative learning control protocols are presented for nonlinear MASs in \cite{Meng2021}.
Nevertheless, all these strategies are based on continuous communication between agents.
We are inspired to take a step toward developing distributed data-driven STC for unknown MASs. 
The difficulty is twofold: i) how to design distributed STM and associated controller for each agent using solely data; and ii) how to obtain a data-driven criterion to guarantee asymptotic consensus of MASs.

This paper aims to design data-driven STC strategies for unknown leader-following MASs using only some state-input data. 
{We begin by establishing a data-driven parametrization for lifted MASs, which allows us to derive data-driven STMs from model-based STMs} to predict the next triggering time without knowing system matrices.
A data-based stability condition is further derived for the closed-loop MAS, in the form of linear matrix inequalities (LMIs), guaranteeing that all followers track the leader asymptotically. 
This condition also provides a solution for designing the matrices of feedback controller and STM using pre-collected noisy state-input data.
Finally, a comparison between data-driven STC, classical model-based STC (that guarantees $\mathcal{H}_{\infty}$-consensus), and system identification-based STC is displayed to validate the effectiveness of the proposed data-driven approach.

In summary, the contributions are as follows.
\begin{enumerate}
\item Building on a data-driven representation of  lifted MASs, a distributed data-driven STM is proposed, where the next triggering time is predetermined using noisy state-input data collected offline and locally.
\item A data-based stability criterion for co-designing the controller gain and triggering matrix is established for unknown MASs under the STM, while achieving asymptotic leader-following consensus.
\item A model-based condition for $\mathcal{H}_{\infty}$-performance analysis and for the existence of distributed STC that ensures a given $\mathcal{H}_{\infty}$-performance level is developed for leader-following MASs.
\end{enumerate}

The remainder of the paper is structured as follows. Some preliminaries and considered problems, including graph theory, model-based STC, and pre-collected data, are introduced in Section \ref{section2}. The data-driven STC design and analysis are given in Section \ref{section3} along with a model-based stability condition for $\mathcal{H}_{\infty}$-consensus control. Section \ref{section4} validates the results by numerical simulations. Section \ref{section5} concludes the paper.

\emph{Notation.} Let $\mathbb{N}$ ($\mathbb{R}$) denote the set of all non-negative integers (real numbers).
 For integers $a< b$, let $\mathbb{N}_{[a,b]}:=\mathbb{N}\cap [a,b]$. 
For a vector $x\in \mathbb{R}^n$, $x>0$ is understood entrywise. {Let $\otimes$ and $*$ denote the Kronecker product and the symmetric part in symmetric matrices.} For a symmetric matrix $P$, $P\succ 0$ ($P\succeq 0$) means that $P$ is positive (semi-)definite. ${\rm{Sym}}\{P\}$ takes the sum of $P^\top$ and $P$. We use $I$ ($\mathbf{0}$) to denote the identity (zero) matrix of appropriate dimensions. {Finally, we write $[\cdot]$ if the elements can be inferred by symmetry.}

\section{Preliminaries and Problem Formulation}
\label{section2}
\subsection{Graph theory}
A weighted graph $\bar{\mathcal{G}}=(\bar{\mathcal{V}}, \bar{\mathcal{E}})$ defined over a  nonempty node set $\bar{\mathcal{V}}=\{v_0, v_1,\ldots, v_N \}$ and an edge set $\bar{\mathcal{E}} \subseteq \bar{\mathcal{V}}\times \bar{\mathcal{V}}$ describes the interaction between agents. 
If node $v_j$ can obtain node $v_i$'s information, then the link $(v_i, v_j)\in \bar{\mathcal{E}}$. 
Let $\mathcal{N}_i=\{j\in \bar{\mathcal{V}}|(i,j)\in \bar{\mathcal{E}}\}$ denote the neighbor set of node $v_i$. 
The induced subgraph $\mathcal{G}=(\mathcal{V},\mathcal{E})$ obtained from $\bar{\mathcal{G}}$ represents the information interaction relationship between followers, where $V= \{v_1,\ldots, v_N\}$ and $\mathcal{E}\subseteq (\mathcal{V}\times\mathcal{V})$. 
The adjacency matrix $\mathcal{A}=[a_{ij}]\in \mathbb{R}^{N\times N}$ has $a_{ii}=0$, $a_{ij}>0$ if $(v_j, v_i)\in \mathcal{E}$, and $a_{ij}=0$, otherwise. 
The Laplacian matrix $\mathcal{L}=[l_{ij}]\in \mathbb{R}^{N\times N}$ associated with the subgraph $\mathcal{G}$ has $l_{ii}=\sum_{j\neq i} {a_{ij}}$ and $l_{ij}=-a_{ij}$.
Let the diagonal matrix $\mathcal{P} = {\rm{diag}}\{a_{10}, \ldots, a_{N0}\}\in \mathbb{R}^{N\times N}$ describe the accessibility of the leader to followers. Precisely, $a_{0i} > 0$ if follower $i$ has access to the information of the leader, and $a_{0i} = 0$ otherwise. In addition, define the matrix $\mathcal{H} := \mathcal{L} + \mathcal{P}$. 
\begin{lemma}[Symmetric positive definite \cite{Zhang2015constructing}]
	If the subgraph $\mathcal{G}$ is undirected and the augmented graph associated with $\mathcal{H}$ is connected, then the matrix $\mathcal{H}$ is symmetric positive definite.
\end{lemma}

\subsection{Distributed self-triggered control}
Consider a linear discrete-time leader-following MAS with  a leader indexed by $0$ and $N$ followers by $1,2,\ldots,N$. For $t\in \mathbb{N}$, the dynamics of the MAS is described by
\begin{equation}
\label{mas}
	\left\{\begin{aligned}{x}_{i}(t+1)&={A}_{\rm tr} x_{i}(t)+{B}_{\rm tr} u_{i}(t), \quad i=1,2,\ldots,N,
		\\ {x}_{0}(t+1)&={A}_{\rm tr} x_{0}(t)\end{aligned}\right.
\end{equation}
where $x_0(t)\in \mathbb{R}^n$ denotes the leader's state, $x_i(t)\in \mathbb{R}^n$ and $u_i(t)\in \mathbb{R}^p$ denote follower $i$'s state and control input, respectively.
We assume that the system matrices ${A}_{\rm tr},{B}_{\rm tr}$ are real, constant, stabilizable, but \emph{unknown}. 

%
We consider the leader-following consensus problem of the unknown MAS \eqref{mas} over a network described by a fixed directed graph $\bar{\mathcal{G}}$.
Our goal is to design self-triggered feedback controllers 
such that the followers can track the leader asymptotically for any given initial state, i.e., $\lim_{t\to \infty}\|x_i(t)-x_0(t)\|=0, \, \forall i=1,2,\ldots,N.$
Before moving on, the following assumption is presented.


\begin{assumption}[Communication topology]
	\label{graph}
 Only a subset of followers have direct access to the leader's information, and the graph $\bar{\mathcal{G}}$ contains a directed spanning tree with the leader as the root. 
\end{assumption}

Let us define the combined measurement variable
\begin{equation}
\label{z}
z_i(t):=\sum_{j=1}^{N} a_{i j}\left({x}_{i}(t)-{x}_{j}(t)\right)+a_{i0}\left({x}_{i}(t)-{x}_{0}(t)\right)
\end{equation}
where $a_{i j}$ denotes the $i j$th entry of the adjacency matrix $\mathcal{A}$. 
We consider the following distributed self-triggered controller for MAS \eqref{mas}
\begin{equation}
\label{controller}
		u_{i}(t)=K z_i(t_k^i),\quad t\in \mathbb{N}_{[t_k^i,\,t_{k+1}^{i}-1]}
\end{equation}
where $K\in \mathbb{R}^{p\times n}$ is the feedback gain matrix to be designed.
Note that each follower $i$ updates the control input \eqref{controller} locally and broadcasts its latest state to neighbors at intermittent instants $\{t_k^i\}_{k=1}^{\infty}$, which are called triggering times, dictated by a self-triggering policy. 
Without loss of generality, we assume that $\bar{x}_0(t)=x_0(t_k^i)=x_0(t)$, that is, the leader does not implement a triggering policy, and each follower $i$ makes a triggered transmission at $t=0$ so that $t_1^i=0$. 
Then, letting $t_k^i$ designate the most recent triggering time, the subsequent triggering time $t_{k+1}^i$  is given by     
\begin{equation}
\label{trigger}
t_{k+1}^i=t_k^i+\inf_{s_k^i\in \mathbb{N}} \left\{s_k^i\geq 1\big|f_i\left(x_i(t_k^i),\,z_i(t_k^i),\, s_k^i\right)\geq0 \right\}
\end{equation}
with the following triggering function 
\begin{equation*}
	f_i\left(x_i(t_k^i),\,z_i(t_k^i),\,s_k^i\right)=e_{i}^{\top} (s_k^i) \Phi e_{i}(s_k^i)-\sigma z_i^{\top}(t_k^i) \Phi z_{i}(t_k^i)
\end{equation*}
where $\Phi\in \mathbb{R}^{n\times n}$ is a positive definite matrix to be designed, $\sigma$ is a positive constant, $s_k^i=t_{k+1}^i-t_k^i$ is the $k$th inter-event triggering interval of agent $i$, and $e_{i}(s_k^i)=x_i(t_k^i+s_k^i)-x_i(t_k^i)$ denotes the error between the last broadcast state at $t_k^i$ and the state at $t_{k+1}^i$. 
Regarding the STM \eqref{trigger}, the following two observations can be made: i) the inter-event time $s_k^i$ can be computed at the current triggering time $t_k^i$, that is, 
each agent decides the next triggering time $t_{k+1}^i=t_k^i+s_k^i$ based on the most recent triggered states from itself and its neighbors; and,
ii) the STM yields asynchronous and intermittent control updates,  transmissions, and state monitoring for all agents.


We define the tracking error $\delta_i(t):=x_i(t)-x_0(t)$ between the leader and each follower $i$. It follows from \eqref{mas} and \eqref{controller} that the dynamics of $\delta_i(t)$ satisfies 
\begin{equation}
\label{delta}
	\delta_i(t+1)={A}_{\rm tr}\delta_i(t)+{B}_{\rm tr}Kz_i(t_k^i), \quad  i=1,2,\ldots,N.
\end{equation}
It becomes self-evident that the leader-following consensus problem of the MAS \eqref{mas} is translated into the stabilization problem of the closed-loop system \eqref{delta}. 

We refer to the controller \eqref{controller} and the STM \eqref{trigger} collectively as a distributed model-based STC law since (pre)determination of communication and control actions requires the precise knowledge of the true system matrices $A_{\rm tr},{B}_{\rm tr}$.
In the literature, a variety of solutions to the model-based STC problem have been provided, see {e.g., \cite{Yi2019,Cui2023,Zhang2022self,Zegers2022,Mi2020}}.
However, when $A_{\rm tr},{B}_{\rm tr}$ are unknown, one cannot apply these results. The challenge left is to design a distributed STC scheme directly from data.
while the data-driven STC rely on the extensively studied model-based STC, the challenge left is to replace process models process with noisy data.

\subsection{Pre-collected data}
\label{section2c}
The absence of matrices ${A}_{\rm tr},{B}_{\rm tr}$ challenges the STC design and associated stability analysis. To address this issue, we collect state-input data $\{x_i(T)\}_{T=0}^{\rho}$, $\{u_i(T)\}_{T=0}^{\rho-1}$  for each follower $i$ across time interval $T\in \{0,1,\ldots,\rho\}$ from the following \emph{perturbed} system
\begin{equation}
\label{mas_noise}
	x_i(T+1)={A}_{\rm tr}x_i(T)+{B}_{\rm tr}u_i(T)+Ew_i(T)
\end{equation}
where $E\in \mathbb{R}^{n\times n_w}$ is a known matrix, {which has full column rank and describes} the influence of \emph{unknown} disturbance $w_i(T)\in\mathbb{R}^{n_w}$ on the system.
{In particular, we assume that we are given noiseless state-input measurements of the leader.
Then, it} follows from \eqref{mas} and the definition of $\delta_i(t)$ that \eqref{mas_noise} can be equivalently converted into the ensuing linear system
\begin{equation}
\label{track}
	\delta_i(T+1)={A}_{\rm tr}\delta_i(T)+{B}_{\rm tr}u_i(T)+Ew_i(T).
\end{equation}

The collected state-input data are corrupted by the \emph{unknown} disturbance $\{w_i(T)\}_{T=0}^{\rho-1}$, which we store in the matrix
$W_i^1 :=\left[w_i(0)\; w_i(1) \;\cdots \; w_i(\rho-1)\right].$
Capitalizing on the noise bound for single systems in \cite{Wang2021,Waarde2022,Wang2022self}, a  quadratic full-block noise bound is presented for the leader-following MASs, along with a general data-driven representation.
\begin{assumption}[Noise bound]
\label{noiseass}
	The noise matrix $W_i$ per follower $i$ belongs to
	\begin{equation}\label{eq:noise}
\mathcal{W}_i=\bigg\{W_i\in \mathbb{R}^{n_w\times\rho} \bigg|\! \left[\begin{matrix}{W_i}^\top \\ I\end{matrix}\right]^\top \!\left[\begin{matrix}Q_{d}& S_{d} \\ *& R_{d}\end{matrix}\right]\!\left[\begin{matrix}{W_i}^\top \\ I\end{matrix} \right] \succeq 0 \bigg\}
\end{equation}
with some known matrices $Q_{d}\prec 0$, $R_{d}=R_{d}^{\top}$, and $S_{d}$ of suitable dimensions. 
\end{assumption}

	\begin{remark}[\emph{Interrelationship between noise models}]
		Assumption \ref{noiseass} provides a comprehensive framework for modeling bounded additive noise.
		Notably, the noise model employed for agent $i$ in Assumption \ref{noiseass} can be regarded as equivalent to the model presented in \cite[Assumption 1]{Waarde2022}, since \eqref{eq:noise} can be transformed into  $ \left[\begin{matrix}I \\ {W_i}^\top \end{matrix}\right]^\top \left[\begin{matrix}R_{d}& S_{d}^\top \\ S_{d}& Q_{d}\end{matrix}\right]\left[\begin{matrix}I \\ {W_i}^\top\end{matrix} \right] \succeq 0$, which	is of the same form as that in \cite{Waarde2022}.
		In addition, the quadratic bound in Assumption \ref{noiseass} can be cast as a specific case of the bound in \cite[Asssumption 2]{Wildhagen2022}.
	\end{remark}

Note that the correspondence between the state-input data $\{x_i(T)\}_{T=0}^{\rho}$, $\{u_i(T)\}_{T=0}^{\rho-1}$ and the true system model is not unique in general, in the sense that there are a family of systems (i.e., matrix pairs) that can generate the collected state-input trajectories,
particularly in light of the fact that the noise sequence corrupting the experiment is unknown.
We compute and collect the tracking error samples $\{\delta_i(T)\}_{T=0}^{\rho}$. 
Based on \eqref{track} and Assumption~\ref{noiseass}, the following result characterizes all such system matrices by a quadratic matrix inequality (QMI).
\begin{lemma}[Data-driven representation of MASs] 
\label{generallem}
Stack the samples $\{\delta_i(T)\}_{T=0}^{\rho}$ and $\{u_i(T)\}_{T=0}^{\rho-1}$ per agent in $\Delta_i :=\left[\delta_i(0) \; \delta_i(1) \; \cdots \; \delta_i(\rho-1)\right],$
$\Delta_{i+} :=\left[\delta_i(1) \; \delta_i(2) \; \cdots \; \delta_i(\rho)\right],$ and $U_i :=\left[u_i(0) \; u_i(1) \; \cdots \; u_i(\rho-1)\right].$
The set $\Sigma_i$ of system matrices interpreting the data of agent $i$ is 
\begin{equation}\label{eq:rep}
	\begin{split}
	\Sigma_i:=\bigg\{[A\; B] \in\, &  \mathbb{R}^{n \times(n+p)} \Big| \\ &\quad \left[\begin{matrix}{[A\; B]^\top } \\ I\end{matrix}\right]^\top  \Theta_i \left[\begin{matrix}{[A\; B]^\top } \\ I\end{matrix}\right] \succeq 0\bigg\}
	\end{split}
\end{equation}
where $\Theta_{i}:=\left[\begin{matrix}-\Delta_i & \mathbf{0} \\ -U_i & \mathbf{0} \\ \hline \Delta_{i+} & E\end{matrix}\right]\left[\begin{matrix}Q_{d} & S_{d} \\ * & R_{d}\end{matrix}\right]\left[\begin{matrix} \cdot\end{matrix} \right]^\top $.
\end{lemma}
	\begin{proof}
		All system matrices $[A\;B]$ that are consistent with the state-input data and the given noise bound satisfy
		$\Delta_{i+}=A\Delta_i+BU_i+EW_i$, which can be rewritten as
		\begin{equation}\label{eq:EW1}
			EW_i=\Delta_{i+}-(A\Delta_i+BU_i).
		\end{equation}
		
		Besides, for any full-rank matrix $E$, it follows from \eqref{eq:noise} that 
		\begin{equation}\label{eq:EW2}
			\left[\begin{matrix}{(EW_i)}^\top \\ E^\top\end{matrix}\right]^\top \left[\begin{matrix}Q_{d}& S_{d} \\ *& R_{d}\end{matrix}\right]\left[\begin{matrix}{(EW_i)}^\top \\ E^\top\end{matrix} \right] \succeq 0.
		\end{equation}
		
		By substituting \eqref{eq:EW1} into \eqref{eq:EW2}, one has 
		\begin{equation}
			\left[\begin{matrix}{[\Delta_{i+}-(A\Delta_i+BU_i)]}^\top \\ E^\top\end{matrix}\right]^\top \left[\begin{matrix}Q_{d}& S_{d} \\ *& R_{d}\end{matrix}\right]\left[\begin{matrix}\cdot\end{matrix} \right]\succeq 0
		\end{equation}
		which implies
		\begin{equation*}
			\Bigg[\left[\begin{matrix}[A\;B]\left[\begin{matrix}-\Delta_i\\-U_i\end{matrix}\right] &\mathbf{0}\end{matrix}\right]+\left[\begin{matrix}\Delta_{i+}&E\end{matrix}\right]\Bigg] \left[\begin{matrix}Q_{d}& S_{d} \\ *& R_{d}\end{matrix}\right]\left[\begin{matrix}\cdot\end{matrix} \right]^\top  \succeq 0.
		\end{equation*}
		Then, we re-express the above inequality as 
		\begin{equation}
			\left[\begin{matrix}\left[\begin{matrix}A&B&I\end{matrix}\right]\left[\begin{matrix}\left[\begin{matrix}-\Delta_i\\-U_i\end{matrix}\right] &\mathbf{0}\\\Delta_{i+}&E\end{matrix}\right]\end{matrix}\right] \left[\begin{matrix}Q_{d}& S_{d} \\ *& R_{d}\end{matrix}\right]\left[\begin{matrix}\cdot\end{matrix} \right]^\top  \succeq 0,
		\end{equation}
		which is equivalent to \eqref{eq:rep},
		completing the proof.
	\end{proof}

It is worth remarking that the state-input trajectories $\{x_i(T)\}_{T=0}^{\rho}$, $\{u_i(T)\}_{T=0}^{\rho-1}$ can be collected  by stimulating the open-loop system \eqref{mas_noise} with designed control inputs offline.
On the other hand, our control objective, i.e., leader-following consensus, is achieved in closed-loop, an online process that is separated from and does not interfere with the offline, open-loop experiment.

\begin{figure}[t]
	\centering
	\includegraphics[scale=0.4]{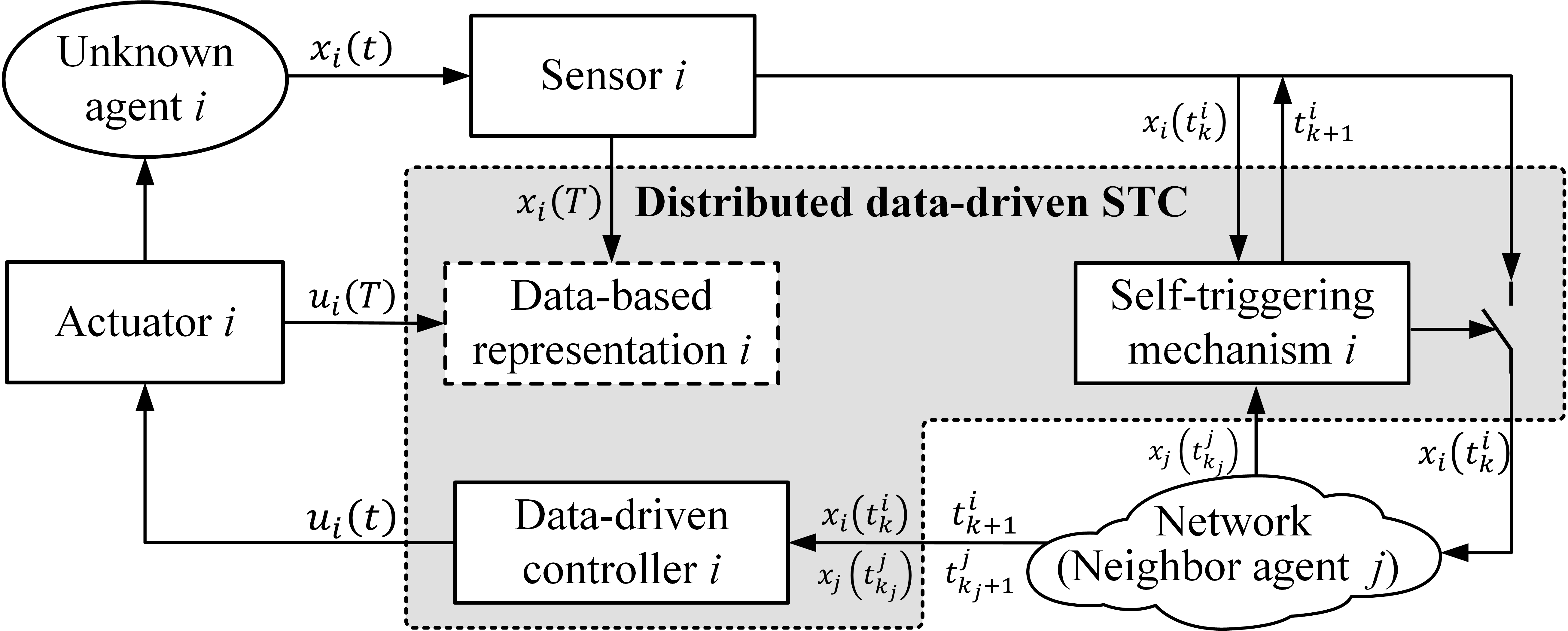}
	\caption{{Distributed data-driven STC for leader-following MAS consensus.}}
	\label{flow}
\end{figure}

\section{Distributed Data-driven Self-triggered Control}
\label{section3}
This section develops a distributed data-driven STC strategy to achieve leader-following consensus  of the unknown MAS in \eqref{mas}. See {Fig.~\ref{flow}} for an illustration of the considered data-driven consensus control architecture.


{To that end, we begin by revisiting the solutions from model-based design to provide a new data-driven STM and stabilization criterion, which arises from matrix manipulations of a model-based preliminary result.}
The states of the closed-loop system \eqref{delta} at triggering instants are governed by the following switched system
\begin{equation}
\label{lift}
	\delta_i(t_k^i+s_k^i)={A}_{\rm tr}^{s_k^i}\delta_i(t_k^i)+\underline{B}_{\rm tr}^{s_k^i}\underline{K}^{s_k^i}z_i(t_k^i)
\end{equation}
where $\underline{K}^{s_k^i}:=\big[ \underbrace{K^\top ~ K^\top ~ \cdots ~ K^\top}_{s_k^i~\text{times}} \big]^\top$ is the lifted controller gain, and ${A}_{\rm tr}^{s_k^i}$ and $\underline{B}_{\rm tr}^{s_k^i}:= [ A_{\rm tr}^{s_k^i-1}B_{\rm tr} ~ A_{\rm tr}^{s_k^i-2}B_{\rm tr} ~ \cdots ~ B_{\rm tr} ]$ are the lifted system matrices with $s_k^i\in \mathbb{N}_{[1,\,\bar{s}]}$ and $1<\bar{s}\in \mathbb{N}$.

Our idea of designing the distributed data-driven STM is unraveled next. 
First, leveraging the switched system \eqref{lift}, the model-based STM \eqref{trigger} can be reformulated using a QMI, which only contains already transmitted states. 
Then, a data-based QMI can be obtained based on a parametrization of the MAS using the S-procedure \cite{Waarde2022}.
However, we acknowledge that translating the model-based STM into the data-driven STM is not straightforward.
The reason is that the unknown lifted matrices $A_{\rm tr}^{s_k^i}$, $\underline{B}_{\rm tr}^{s_k^i}$ in \eqref{lift}, corresponding to different inter-event triggering intervals, prevent us from employing the general data-driven representation in Lem.~\ref{generallem}.
This motivates us to develop a distributed data-driven approach that is applicable to the self-triggered MAS.



Toward this end, we begin by constructing a distributed data-driven STM based on a data-driven representation of lifted MASs in Sec.~\ref{section3b}.
Subsequently, a distributed data-driven STC algorithm is proposed and corroborated by rigorous stability analysis in Sec.~\ref{section3c}. 
Finally, to compare the robustness of data-driven and model-based approaches against noise, we establish a model-based stability condition for guaranteed $\mathcal{H}_\infty$ performance in Sec.~\ref{section3d}.


\subsection{Data-driven STM design}
\label{section3b}

Before formally introducing our data-driven STM, we first propose a data-driven representation for the lifted leader-following MAS, which relies only on pre-collected data and a slightly altered assumption on the noise.

Given data $\{\delta_i(T)\}_{T=0}^{\rho+s-1}$, $\{u_i(T)\}_{T=0}^{\rho+s-2}$ per agent, we can construct the following matrices for all $s\in \mathbb{N}_{[1,\bar{s}]}$.
\begin{align*}
\Delta_{i+}^s&:=\left[\begin{matrix}\delta_i(s) \quad \delta_i(s+1) \quad \dots \quad \delta_i(\rho+s-1)\end{matrix}\right],\\
U_i^s&:=\left[\begin{matrix}
	u_i(0) & u_i(1) & \cdots &u_i(\rho-1) \\ \vdots & \vdots &\ddots &\vdots \\  u_i(s-1) & u_i(s) & \cdots &u_i(\rho+s-2)
	\end{matrix}\right].
\end{align*}

Moreover, the lifted noise $W_i^1$ for $s=1$ has been defined in Sec.\ref{section2c}, and that for $s>1$ is similarly given by
\begin{align*}
W^s_i :=\left[ A_{\rm tr}^{s-1}E ~ A_{\rm tr}^{s-2}E ~ \cdots ~ E \right]\underline{W}_i^s,\quad s\in \mathbb{N}_{[2,\bar{s}]}
\end{align*}
where
$$\underline{W}_i^s:=\left[\begin{matrix}
	w_i(0) & w_i(1) & \cdots &w_i(\rho-1) \\ \vdots & \vdots &\ddots &\vdots \\  w_i(s-1) & w_i(s) & \cdots &w_i(\rho+s-2)
	\end{matrix}\right].$$


Inspired by \cite{Wildhagen2022}, the following assumption on the noise $W^s_i$ of follower $i$ is proposed in terms of a QMI.
\begin{assumption}[Lifted noise bound]
\label{ass_lift_noise}
	The noise sequence $\{w_i(T)\}_{T=0}^{\rho+s-2}$ of follower $i$ stored in  $W_i^s$ satisfies the bound
\begin{equation}
\label{lift_noise}
\mathcal{W}_i^s=\Big\{W_i^s\in \mathbb{R}^{n_w^s\times\rho} \Big| \left[\begin{matrix} {W_i^s}^\top \\ I\end{matrix}\right]^\top\left[\begin{matrix} Q_{d}^s& S_{d}^s \\ *& R_{d}^s\end{matrix}\right]\left[\begin{matrix} {W_i^s}^\top \\ I\end{matrix}\right] \succeq 0 \Big\}
\end{equation}
with prescribed matrices $S_{d}^s\in \mathbb{R}^{\rho\times n_w^s}$, $R_{d}^s=R_{d}^{s\top}\in \mathbb{R}^{n_w^s\times n_w^s}$, $Q_{d}^s\in \mathbb{R}^{\rho\times n_w^s}\prec 0$, where $n_w^1=n_w$ and $n_w^s=n$ for $s\in \mathbb{N}_{[2,\bar{s}]}$.
\end{assumption}


Then, let us define the set of all lifted matrices $A^s$, $\underline{B}^s$ that explains the data $(\Delta_{i+}^s, \Delta_i, U_i^s, W_i^s)$ of agent $i$ as
\begin{equation}
\label{perturb}
	\Sigma_i^s:=\left\{[A^s\; \underline{B}^s] \mid \Delta_{i+}^s=A^s \Delta_i+\underline{B}^s U_i^s+E^s W_i^s\right\}
\end{equation}
where $E^1:=E$ for $s=1$ and $E^s:=I$ for $s\in \mathbb{N}_{[2,\bar{s}]}$.

{Before proceeding, a prerequisite assumption is posed. 

\begin{assumption}[Requirement of data]
	\label{ass_Theta}
	The matrix 
	\begin{equation}\label{eq:rank}
		\Theta_{i}^s:=\left[\begin{array}{cc}-\Delta_i & \mathbf{0} \\ -U_i^s & \mathbf{0} \\ \hline \Delta_{i+}^s & E^s\end{array}\right]\left[\begin{array}{cc}Q_{d}^s & S_{d}^s \\ * & R_{d}^s\end{array}\right]\left[\begin{matrix} \cdot\end{matrix} \right]^\top 
	\end{equation} has full column rank.
\end{assumption}}

Assumption~\ref{ass_Theta} is standard for stability analysis of linear systems. It has been shown in \cite{Berberich2021ap} that when the pre-collected data are sufficiently rich and $E$ is invertible, Assumption~\ref{ass_Theta} is satisfied.
{Based on the assumption above, we generalize the data-driven parameterization of single systems provided in \cite{Wildhagen2022} and advocate the following data-driven representation of lifted MASs.}

{\begin{lemma}[Data-driven representation of lifted MASs]
\label{representation}
Suppose Assumptions \ref{ass_lift_noise} and \ref{ass_Theta} hold.
The set $\bar{\Sigma}_i^s$ of agent $i$ can be expressed by the QMI 
\begin{equation}
	\label{eq:data1}
	\begin{split}
		\bar{\Sigma}_i^s:=&\bigg\{[A^s\; \underline{B}^s] \in\,   \mathbb{R}^{n \times(n+sp)} \Big| \\ &\quad \left[\begin{matrix}{[A^s\; \underline{B}^s] } \\ I\\I\end{matrix}\right]^\top  \left[\begin{matrix}\hat{\Theta}_i^s & \mathbf{0}\\ \mathbf{0} & M\end{matrix}\right]\left[\begin{matrix}{[A^s\; \underline{B}^s]} \\ I\\I\end{matrix}\right]\succeq 0\bigg\}
	\end{split}
\end{equation}
with $M\succ0$, $	\hat{\Theta}_i^s\!:=\!\left[\begin{matrix}-\hat{R}_{d}^s & \hat{S}_{d}^{s\top} \\ * & -\hat{Q}_{d}^s\end{matrix}\right] $, and	$\left[\begin{matrix}\hat{Q}_{d}^s & \hat{S}_{d}^{s\top} \\ * & \hat{R}_{d}^s\end{matrix}\right]:=\left[\begin{matrix}Q_{d}^s & S_{d}^s \\ * & R_{d}^s\end{matrix}\right]^{-1}$.
\end{lemma}}

\begin{proof}
	Similarly to the proof of Lem. \ref{generallem}, the set $\Sigma_i^s$ in \eqref{perturb} can be transformed into a quadratic constraint as below
	\begin{equation}
		\label{Theta}
		\begin{split}
			\Sigma_i^s:=&\bigg\{[A^s\; \underline{B}^s] \in\,   \mathbb{R}^{n \times(n+sp)} \Big| \\ &\quad \left[\begin{matrix}{[A^s\; \underline{B}^s]^\top } \\ I\end{matrix}\right]^\top  \Theta_i^s \left[\begin{matrix}{[A^s\; \underline{B}^s]^\top } \\ I\end{matrix}\right] \succeq 0\bigg\},
		\end{split}
	\end{equation}
	where $\Theta_i^s$ has been defined in \eqref{eq:rank}.
	
	In addition, it follows from \eqref{Theta} and the dualization lemma \cite{Scherer2000} that the following inequality holds for all $[A^s\; \underline{B}^s]\in {\Sigma}_i^s$
	\begin{equation}\label{eq:data0}
		\left[\begin{matrix}{[A^s\; \underline{B}^s] } \\ I\end{matrix}\right]^\top  \hat{\Theta}_i^s \left[\begin{matrix}{[A^s\; \underline{B}^s]} \\ I\end{matrix}\right] \succeq 0.
	\end{equation}

For any given matrix $M\succ 0$, a modified data-driven representation of lifted MASs can be expressed by the set $\bar{\Sigma}_i^s$ defined in \eqref{eq:data1}.
Due to the introduction of $M$, any $[A^s\; \underline{B}^s]$ satisfying \eqref{eq:data0} also satisfies \eqref{eq:data1}, which implies ${\Sigma}_i^s\subseteq \bar{\Sigma}_i^s$. 
Hence, the set $\bar{\Sigma}_i^s$ contains all possible  matrices $[A^s\; \underline{B}^s]$ compatible with the measured data and the lifted noise.
\end{proof}


\begin{figure*}[t]
	\begin{equation}
		\label{long}
		\tag{FQ}
		\begin{split}
			\mathcal{F}_i(\delta_i(t_k^i),\, z_i(t_k^i))&:=\left[\begin{matrix}
				I & \mathbf{0} & \mathbf{0} \\ \mathbf{0} & \delta_i(t_k^i) &\mathbf{0} \\   \mathbf{0} & \mathbf{0} & z_i(t_k^i)
			\end{matrix}\right]^\top
			{\left[\begin{matrix}
					-\Phi & \Phi & \mathbf{0} \\ \Phi & -\Phi &\mathbf{0} \\   \mathbf{0} & \mathbf{0} & \sigma\Phi
				\end{matrix}\right]}\left[\begin{matrix}\cdot\end{matrix}\right],\\
			\mathcal{Q}_i^s(\delta_i(t_k^i),\,z_i(t_k^i))&:=\left[\begin{matrix}
				I & \mathbf{0} & \mathbf{0} &\mathbf{0}&\mathbf{0}\\ \mathbf{0} & \delta_i^\top(t_k^i) &z_i^\top(t_k^i) {(\underline{K}^s)^\top} &\mathbf{0}&\mathbf{0}  \\  \mathbf{0}& \mathbf{0} & \mathbf{0} & \delta_i^\top(t_k^i) & z_i^\top(t_k^i) {(\underline{K}^s)^\top}	\end{matrix}\right]\! \left[\begin{matrix}\hat{\Theta}_i^s & \mathbf{0}\\ \mathbf{0} & M\end{matrix}\right]\!
			\left[\begin{matrix} \cdot \end{matrix}  \right]^\top.
		\end{split}
	\end{equation}
	\rule[-10pt]{17.9cm}{0.08em}
\end{figure*}

\begin{remark}\label{rmk:representation}
	[\emph{Discussion on the lifted data-driven representation}] 
It is worth emphasizing that we transform the original set ${\Sigma}_i^s$ of allowable system matrices into $\bar{\Sigma}_i^s$. The rationale behind this stems from the fact that, unlike single systems, an additional term $z_i(t)$ involved in the STC protocol \eqref{controller} and \eqref{trigger}, which accounts for local interactions between agents, is necessary for achieving consensus of MASs. 
In line with this, we introduce the matrix $M$ to construct a agent-wise data-driven representation of the lifted MAS in Lem. \ref{representation}, which provides an avenue to develop a distributed STM employing data rather than explicit models. 
However, it is not straightforward how one can apply Lem. \ref{representation} in the data-driven STM design. 
This is mainly due to the difficulty of obtaining a strict characterization as in \eqref{lift_noise} for $s\in \mathbb{N}_{[2,\bar{s}]}$, since the unknown lifted sequence $A^s$ is coupled with the lifted noise $\mathcal{W}_i^s$.
Thus, it is necessary to estimate an upper bound on the set $\mathcal{W}_i^s$ and compute $\Theta_i^s$ for all $s\in \mathbb{N}_{[1,\bar{s}]}$.
We recall from \cite{Wildhagen2022} that an overapproximation of $\mathcal{W}_i^s$ can be derived by estimating the maximum singular value of $A^s$ using only measured data in advance.
Likewise, we call for \cite[Alg. 1]{Wildhagen2022} to iteratively compute the lifted noise bound and lifted system representations from $s=1$ to $\bar{s}$ for each follower in the network.
\end{remark}


The next lemma provides an equivalent model-based triggering condition for all possible $[A^s\; \underline{B}^s]\in \bar{\Sigma}_i^s~(s\in\mathbb{N}_{[1,\bar{s}]})$.

\begin{lemma}[QMI-form model-based triggering condition]
	\label{lem_tri} 
	 Suppose Assumptions \ref{ass_lift_noise} and \ref{ass_Theta} hold.
	The model-based STM \eqref{trigger} is satisfied if the following condition holds for all $[A^s\; \underline{B}^s]\in \bar{\Sigma}_i^s~(s\in\mathbb{N}_{[1,\bar{s}]})$
	\begin{equation}
		\label{trigger_AB}
		\left[\begin{matrix}{A^{s}\delta_i(t_k^i)+\underline{B}^{s}\underline{K}^{s}z_i(t_k^i)}\\ \delta_i(t_k^i) \\z_i(t_k^i)\end{matrix}\right]^\top \!
		\left[\begin{matrix}
			-\Phi & \Phi & \mathbf{0} \\ \Phi & -\Phi &\mathbf{0} \\   \mathbf{0} & \mathbf{0} & \sigma\Phi
		\end{matrix}\right]\!
		\left[\begin{matrix}\cdot\end{matrix}\right]\succeq 0.
	\end{equation}
\end{lemma}

\begin{proof}
	The model-based STM \eqref{trigger} can be equivalently
	translated into the following version
	\begin{equation}
		\label{trigger_dual}
		t_{k+1}^i=t_k^i+\max_{s_k^i\in \mathbb{N}}\left\{s_k^i\geq 1\big|f_i^{'}\left(x_i(t_k^i),\,z_i(t_k^i),\, s_k^i\right)\geq0 \right\}
	\end{equation}
where $f_i^{'}\left(x_i(t_k^i),\,z_i(t_k^i),\, s_k^i\right)=-f_i\left(x_i(t_k^i),\,z_i(t_k^i),\, s_k^i\right)$.

Since $\delta_i(t_k^i)=x_i(t_k^i)-x_0(t)$, one has that $e_{i}(s_k^i)=\delta_i(t_k^i+s_k^i)-\delta_i(t_k^i)$. 
By means of \eqref{lift}, the triggering condition $f_i^{'}\left(x_i(t_k^i),\,z_i(t_k^i),\,s_k^i\right)\geq0$ is rewritten as the following QMI
\begin{equation}
	\label{self2}
	\left[\begin{matrix}{A_{\rm tr}^{s_k^i}\delta_i(t_k^i)+\underline{B}_{\rm tr}^{s_k^i}\underline{K}^{s_k^i}z_i(t_k^i)}\\ \delta_i(t_k^i) \\z_i(t_k^i)\end{matrix}\right]^\top  
	\left[\begin{matrix}
		-\Phi & \Phi & \mathbf{0} \\ \Phi & -\Phi &\mathbf{0} \\   \mathbf{0} & \mathbf{0} & \sigma\Phi
	\end{matrix}\right]
	\left[\begin{matrix}\cdot\end{matrix}\right]\succeq 0.
\end{equation}

Furthermore, since the true system matrices $[A_{\rm tr}^{s_k^i}\; \underline{B}_{\rm tr}^{s_k^i}]\in\Sigma_i^s\subseteq \bar{\Sigma}_i^s$, the condition \eqref{self2} holds for $[A_{\rm tr}^{s_k^i}\; \underline{B}_{\rm tr}^{s_k^i}]$ if it holds true for all $[A^s\; \underline{B}^s]\in \bar{\Sigma}_i^s~(s\in\mathbb{N}_{[1,\bar{s}]})$. 
That is, the model-based condition \eqref{self2} is guaranteed by the condition \eqref{trigger_AB} for all $[A^s\; \underline{B}^s]\in \bar{\Sigma}_i^s$, which completes the proof.
\end{proof}

Having obtained the data-driven representation in Lem.~\ref{representation} and the model-based triggering condition for all $[A^s\; \underline{B}^s]\in \bar{\Sigma}_i^s$ in Lem.~\ref{lem_tri}, a data-driven STM will be built in the subsequent.

\begin{thm}[Data-driven self-triggering condition]
\label{thm1}
	 Consider the MAS \eqref{mas} and the feedback controller  \eqref{controller} under graph $\bar{\mathcal{G}}$. 
	 Suppose Assumptions \ref{ass_lift_noise}-\ref{ass_Theta} hold. For a given scalar $\sigma>0$, a controller gain $K$, a triggering matrix $\Phi\succ 0$, and the latest transmitted state $x_i(t_k^i)$ of agent $i$, if and only if there exists a scalar $\alpha>0$ such that the following LMI holds for $i=1,2,\ldots,N$ and some $s\in\mathbb{N}_{[1,\bar{s}]}$
\begin{align}
\label{pf1}
	\mathcal{F}_i(\delta_i(t_k^i),\,z_i(t_k^i))-\alpha \mathcal{Q}_i^s(\delta_i(t_k^i),\,z_i(t_k^i))\succeq0,
\end{align}
 the triggering condition \eqref{trigger_AB} is satisfied for any {$[A^s\; \underline{B}^s]\in \bar{\Sigma}_i^s$}, where matrices $\mathcal{F}_i(\delta_i(t_k^i),\,z_i(t_k^i))$ and $\mathcal{Q}_i^s(\delta_i(t_k^i),\,z_i(t_k^i))$ are defined in \eqref{long}.
\end{thm}

\begin{proof}
First, to prove the ``only if" statement, suppose that the model-based triggering condition \eqref{trigger_AB} is satisfied for all $[A^s\; \underline{B}^s]\in\bar{\Sigma}_i^s$, which can be rewritten as
\begin{equation}
	\label{self3}
	\left[\begin{matrix} {A^{s}\delta_i(t_k^i)+\underline{B}^{s}\underline{K}^{s}z_i(t_k^i)}\\ I \\I\end{matrix}\right]^\top  
	\mathcal{F}(\delta_i(t_k^i),\, z_i(t_k^i))
	\left[\begin{matrix}\cdot\end{matrix}\right]\succeq 0.
\end{equation}

In addition, let us re-express \eqref{eq:data1}. For any $[A^s\; \underline{B}^s]\in\bar{\Sigma}_i^s$, one has
\begin{equation}
	\label{eq:data2}
	\left[\begin{matrix}A^s\delta_i(t_k^i)+\underline{B}^s\underline{K}^s z_i(t_k^i) \\ I\\I\end{matrix}\right]^\top \!\!\mathcal{Q}^s(\delta_i(t_k^i),\,z_i(t_k^i))
	\left[\begin{matrix}\cdot \end{matrix}\right] \succeq 0.
\end{equation}

Now, the working assumptions of \emph{S}-Lemma, as stated in \cite[Thm. 9]{Waarde2022}, are fulfilled.
By utilizing this full-block S-procedure, we conclude that if the QMI \eqref{self3} holds for any $[A^s\; \underline{B}^s]\in \bar{\Sigma}_i^s~(s\in\mathbb{N}_{[1,\bar{s}]})$, there exists a scalar $\alpha >0$ such that
	$\mathcal{F}_i(\delta_i(t_k^i),\,z_i(t_k^i))-\alpha \mathcal{Q}_i^s(\delta_i(t_k^i),\,z_i(t_k^i))\succeq0.$

Conversely, to prove the ``if" statement, suppose there exists a $\alpha>0$ satisfying \eqref{pf1} for $i=1,2,\ldots,N$. As such, \eqref{self3} and \eqref{eq:data2} hold for any $[A^s\; \underline{B}^s]\in \bar{\Sigma}_i^s~(s\in\mathbb{N}_{[1,\bar{s}]})$. We
conclude that the triggering condition \eqref{trigger_AB} is satisfied for any $[A^s\; \underline{B}^s]\in \bar{\Sigma}_i^s$, which accomplishes the proof.
\end{proof}

According to Thm.~\ref{thm1}, we are in a position to obtain the following data-driven STM reminiscent of \eqref{trigger}, yet eliminating the need for any prior model knowledge
\begin{equation}
\label{trigger_data}
	t_{k+1}^i=t_k^i+\max_{s_k^i\in \mathbb{N}} \left\{s_k^i\geq 1\big|\hat{f}_i\left(x_i(t_k^i),\,z_i(t_k^i),\, s_k^i\right)\succeq 0
 \right\}
\end{equation}
where the data-driven triggering function is given by
\begin{align}
	\hat{f}_i\left(x_i(t_k^i),\,z_i(t_k^i),\, s_k^i\right)&=\mathcal{F}_i(\delta_i(t_k^i),\,z_i(t_k^i))\nonumber \\ 
	&\quad -\alpha \mathcal{Q}_i^s(\delta_i(t_k^i),\,z_i(t_k^i)). \label{self_data}
\end{align}

The data-driven STM \eqref{trigger_data} makes it possible to predetermine the next triggering time $t_{k+1}^i$ at $t_k^i$ using the already transmitted information and pre-collected state-input data $\{x_i(T)\}_{T=0}^{\rho+s-1}$, $\{u_i(T)\}_{T=0}^{\rho+s-2}$.

\begin{remark}[\emph{Relationship between the data-driven STM and the	model-based STM}]
	 Thm. \ref{thm1} implies the data-driven condition \eqref{pf1} and the model-based condition \eqref{trigger_AB} are equivalent for the same set $\bar{\Sigma}_i^s$. 
	 Moreover, the model-based STM \eqref{trigger} only holds for the true system matrices $[A_{\rm tr}^{s_k^i}\; \underline{B}_{\rm tr}^{s_k^i}]$, while it can be inferred from Lem. \ref{lem_tri} and Thm. \ref{thm1} that the data-driven STM \eqref{trigger_data} applies to all $[A^s\; \underline{B}^s]\in \bar{\Sigma}_i^s$ at the price of introducing conservatism.
Consequently, the data-driven STM \eqref{trigger_data} does not result in a larger minimum inter-event triggering interval compared to the model-based one \eqref{trigger} without noise. 
\end{remark}

\begin{remark}[\emph{Novelty}]
The novelty of the distributed data-driven STM lies in three aspects.
First, the distributed setting for MASs is investigated in Thm.~\ref{thm1}, which introduces  local interaction between agents (c.f., $z_i(t)$), posing a challenge to the design of data-driven STM.
Thereby, a data-driven representation of lifted MASs is derived in Lem.~\ref{representation} by introducing a positive matrix $M$.
Using this representation and the S-procedure, we are finally able to translate the model-based condition \eqref{self3} into the data-driven one \eqref{pf1}. 
Second, the data-driven STM \eqref{trigger_data} removes the dependence on explicit system models, which can be more applicable to real-world applications. 
Third, the data-driven STM \eqref{trigger_data} avoids successive monitoring triggering conditions, and significantly
reduces inter-agent communication and energy consumption.
\end{remark}

\begin{remark}[\emph{Flexibility of the controller gain $K$}]
	The choice of controller gain $K$ in Thm. \ref{thm1} is unconstrained. This flexibility arises due to the current focus of Thm. \ref{thm1}, which primarily concentrates on devising a data-driven STM capable of directly predicting the subsequent triggering time from noisy data. The forthcoming subsection will delve into the examination of data-driven stability analysis and design.
\end{remark}

\subsection{Data-driven stability analysis and controller design}
\label{section3c}
Due to the lack of knowledge about the true system matrices $A_{\rm tr}$, $B_{\rm tr}$, it is necessary to ensure stability of the closed-loop system \eqref{delta} for all possible $[A\; B]\in \Sigma_i$.
Based on the data-driven representation of MASs in Lem.~\ref{generallem} and the controller \eqref{controller}, our data-driven self-triggered consensus algorithm for unknown MASs under the STM \eqref{trigger_data}, is presented in Alg.~\ref{Alg}, with stability guarantees provided below.

\begin{algorithm}[!htb]
	\caption{Distributed data-driven self-triggered consensus}
	\label{Alg}
	\begin{algorithmic}[1]
		\State \textbf{Input:} desired lifespan of MAS $T$; triggering number $k=1$; initial states of leader $x_0(0)\in \mathbb{R}^n$ and follower $x_i(0) \in \mathbb{R}^n$; current triggering time $t_k^i$; triggering parameter $\sigma>0$; matrices in the noise model $Q_d$, $S_d$, and $R_d$; {parameters $\sigma$ and $\epsilon$;} matrices $\Theta_i^s$ for all $s=1,2,\ldots,\bar{s}$, and state-input data $\{x_i(T)\}_{T=0}^{\rho}$, $\{u_i(T)\}_{T=0}^{\rho-1}$ for $\forall i=1,2,\ldots,N$.
		\State \textbf{Construct} data matrices $\Delta_{i+}$, $\Delta_i$, and $U_i$.
		\State \textbf{Build} $\Theta_i$ in \eqref{eq:rep} with given matrices $Q_d$, $S_d$, and $R_d$ satisfying Assumption \ref{noiseass}.
		\State \textbf{Search} for feasible matrices $\bar{\Phi}$, $G$, and $K_G$ by solving LMIs in \eqref{eqthm2} with properly chosen parameters $\sigma$ and $\epsilon$.
		\State \textbf{Design} the controller gain $K=K_GG^{-1}$ and the triggering matrix $\Phi=(G^{-1})^\top\bar{\Phi}G^{-1}$.
		\While {$t<T$}
		\For{$i=0,1,2,\ldots, N$}
		\If{$t=t_{k}^i$}
		\label{6}
		\State{Broadcast $x_i(t_k^i)$ to agent $j\in \mathcal{N}_i$;}
		\State {Compute $z_i(t_k^i)$ from \eqref{z} based on the updated $x_j(t_k^j)$, $j\in \mathcal{N}_i$;}
		\State{Update the control protocol \eqref{controller} and the dynamics \eqref{mas} of agent $i$.}
		\If  {$\hat{f}_i\left(x_i(t_k^i),\,z_i(t_k^i),\, s_k^i\right)\succeq 0$ is satisfied}
		\label{9}
		\State {Obtain $s_k^i$ and predict the next triggering time $t_{k+1}^i=t_k^i+s_k^i$;}
		\Else
		\State {Set $s_k^i=s_k^i+1$ and update $\hat{f}_i(t)$ from \eqref{self_data}.}
		\EndIf
		\EndIf
		\State \textbf{Set} $k=k+1$ and go back to \ref{6}.
		\EndFor
		\EndWhile
	\end{algorithmic}
\end{algorithm}

\begin{thm}[Data-driven consensus and design]
\label{thm2}
	Consider the MAS \eqref{mas} under graph $\bar{\mathcal{G}}$,  and let Assumptions~\ref{graph}-\ref{ass_Theta} hold. For given scalars $\sigma>0$ and $\epsilon$, the leader-following consensus is achieved asymptotically for any initial state under the feedback controller \eqref{controller} and the data-driven STM \eqref{trigger_data} for any $[A\; B]\in \Sigma_i$, if there exist some scalar $\beta>0$ and matrices $P\succ0$, ${\bar{\Phi}}\succ0$, $G$, $K_G$ such that the following LMIs are satisfied for $i=1,2,\ldots,N$
\begin{equation}
\label{eqthm2}
	\left[\begin{matrix} \mathbf{0}&\mathcal{T}\\ \ast &\Omega+{\Psi}\end{matrix}\right]+\beta (I_N\otimes {\widetilde{\Theta}_i})\prec 0 
\end{equation}
	where
\begin{align*}
	\mathcal{T}&:=[(I_N\otimes G L_1)^\top,\;({\mathcal{H}}\otimes K_G L_3 )^\top]^\top\\
	\mathcal{R}&:=(L_1+\epsilon L_2)^\top,\quad
	{\widetilde{\Theta}_i:=\left[\begin{matrix} I & \mathbf{0}\\\mathbf{0}& \mathcal{R}\end{matrix}\right]
	{\Theta}_i\left[\begin{matrix} I & \mathbf{0}\\\mathbf{0}& \mathcal{R}\end{matrix}\right]^\top}\\
	\Omega &:=L_2^\top(I_N\otimes P)L_2-L_1^\top(I_N\otimes P)L_1\\
	\Psi &:={\rm Sym}\{-(I_N\otimes \mathcal{R} GL_2)\}+\sigma(L_3^\top ({\mathcal{H}}\otimes {\bar{\Phi}})L_3)\\
	&~\quad-(L_3-L_1)^\top(I_N\otimes {\bar{\Phi}} )(L_3-L_1)\\
	L_\kappa&:=\big[\mathbf{0}_{n\times (\kappa-1)n},\, I_n,\, \mathbf{0}_{n\times (3-\kappa)n}\big],\;\kappa=1,\,2,\,3\\
	\end{align*}
Moreover, the controller gain is given by $K=K_{G}G^{-1}$ {and the triggering matrix is co-designed as $\Phi={(G^{-1})}^\top\bar{\Phi}G^{-1}$.}
\end{thm}



\begin{proof}
First, let $\delta_i(t) = G s_i(t)$, where $G\in \mathbb{R}^{n\times n}$ is assumed nonsingular. For $t\in \mathbb{N}_{[t_k^i,~t_{k+1}^i-1]}$, the dynamics of uncertain tracking error system can be reconstructed as
\begin{equation}
\label{s}
	s(t+1)=(I_N\otimes G^{-1}AG)s(t) +(\mathcal{H}\otimes G^{-1}BK_G)\bar{s}(t)
\end{equation}
where $\mathcal{H}=\mathcal{L}+\mathcal{P}$, $s(t)=[s_1^\top (t),\,s_2^\top (t), \ldots,\,s_N^\top (t)]^\top$, $\bar{s}(t)=[\bar{s}_1^\top (t),\,\bar{s}_2^\top (t), \ldots,\,\bar{s}_N^\top (t)]^\top$ with $\bar{s}_i(t)=s_i(t_k^i)$, and $K_{G}=KG$.
The system \eqref{s} exhibits the same characteristics as \eqref{delta} in terms of stability and performance.

Consider the following  Lyapunov candidate function
\begin{equation*}
	V(t)=s^\top(t)(I_N\otimes P)s(t)
\end{equation*}
where $P\succ 0$.
The forward difference $\Delta V(t):=V(t+1)-V(t)$ along the trajectory of \eqref{s} yields that
\begin{equation}
\label{V1}
	\Delta V=\zeta^\top(s,t)\left(L_2^\top(I_N\otimes P)L_2-L_1^\top(I_N\otimes P)L_1 \right)\zeta(s,t)
\end{equation}
with $\zeta(s,t):=[s^\top(t),~s^\top(t+1),~\bar{s}^\top(t)]^\top$.

Applying the descriptor method \cite{Fridman2001}, the system \eqref{s} can be expressed as follows
\begin{align}
\label{V2}
	&~\quad2[s(t)+\epsilon s(t+1)]^\top (I_N\otimes G)[(I_N\otimes G^{-1}AG)s(t)\nonumber \\
	&\quad~~+(\mathcal{H}\otimes G^{-1}BK_G)\bar{s}(t)-s(t)]\nonumber\\ 
	&~~~~~~~=2\zeta(t)^\top (I_N\otimes \mathcal{R})\Big[(I_N\otimes AGL_1)+(\mathcal{H}\otimes BK_GL_3)\nonumber \\ 
&\quad\quad\quad\quad\quad\quad\quad\quad\quad~~~~~~-(I_N\otimes GL_2) \Big] \zeta(t)=0
\end{align}
where $\mathcal{R}=(L_1+\epsilon L_2)^\top$.

Evidently, Thm. \ref{thm1} indicates that our data-driven STM \eqref{trigger_data} guarantees \eqref{trigger}, so for $t\in \mathbb{N}_ {[t_k^i,~t_{k+1}^i-1]}$ one obtains
\begin{equation}
\label{V3}
\begin{split}
	\zeta^\top(s,t)&\Big[-(L_3-L_1)^\top(I_N\otimes {\bar{\Phi}})(L_3-L_1)\\
	&~+\sigma\left(L_3^\top ({\mathcal{H}\otimes \bar{\Phi}})L_3\right) \Big]\zeta(s,t)\geq0.
	\end{split}
\end{equation}

By summing up \eqref{V1}-\eqref{V3}, $\Delta V(t)$ is bounded by
\begin{equation}
\label{V4}
	\Delta V\leq \zeta^\top(s,t)\Upsilon \zeta(s,t),
\end{equation}
where $\Upsilon:=\Omega + \Psi + {\rm Sym}\big\{(I_N\otimes \mathcal{R}AGL_1)+({\mathcal{H}}\otimes \mathcal{R}BK_GL_3)\big\}$.
In addition, $\Upsilon$ can be rewritten as 
\begin{align}
	\Upsilon:=\left[\begin{matrix} {[A_N\; B_N]^\top } \\ I\end{matrix} \right]^\top \left[\begin{matrix} \mathbf{0}&\mathcal{T}\\ \ast &\Omega+\Psi\end{matrix} \right]
	\left[\begin{matrix} {[A_N\; B_N]^\top } \\ I\end{matrix} \right]
\end{align}
with $A_N = I_N\otimes \mathcal{R}A$ and $B_N = I_N\otimes \mathcal{R}B$.

Recall the general data-driven representation of MASs in Lem.~\ref{generallem}. For all $[A\; B]\in \Sigma_i$, it follows that
\begin{equation}
\label{AB}
	\left[\begin{matrix}{[A\; B]^\top } \\ I\end{matrix}\right]^\top  \Theta_i \left[\begin{matrix}{[A\; B]^\top } \\ I\end{matrix}\right] \succeq 0.
\end{equation}
Pre- and post-multiplying \eqref{AB} by $\mathcal{R}$ and $\mathcal{R}^\top$ yield
\begin{equation}\label{eq:rarb}
	\left[\begin{matrix} \mathcal{R}A & \mathcal{R}B&I\end{matrix}\right]	\left[\begin{matrix} I & \mathbf{0}\\\mathbf{0}& \mathcal{R}\end{matrix}\right]
	{\Theta}_i\left[\begin{matrix} I & \mathbf{0}\\\mathbf{0}& \mathcal{R}\end{matrix}\right]^\top
	\left[\begin{matrix} (\mathcal{R}A)^\top \\ (\mathcal{R}B)^\top\\I\end{matrix}\right] \succeq 0.
\end{equation}
Then, for $i=1,2,\ldots,N$, \eqref{eq:rarb} can be written in the following compact form
\begin{equation*}
	\left[\begin{matrix} {[A_N\; B_N]^\top } \\ I\end{matrix}\right]^\top 
	( I_N\otimes\widetilde{\Theta}_i)
	\left[\begin{matrix} {[A_N\; B_N]^\top } \\ I\end{matrix}\right] \succeq 0
	\end{equation*}
where $\widetilde{\Theta}_i:=\left[\begin{matrix} I & \mathbf{0}\\\mathbf{0}& \mathcal{R}\end{matrix}\right]
{\Theta}_i\left[\begin{matrix} I & \mathbf{0}\\\mathbf{0}& \mathcal{R}\end{matrix}\right]^\top$.

Thus, we use the full-block S-procedure again to conclude that the condition $\Upsilon\prec 0$ holds for any $[A\; B]\in \Sigma_i$ if there exists a scalar $\beta > 0$ such that \eqref{eqthm2} holds.
 
The LMI \eqref{eqthm2} provides a sufficient condition for $\Upsilon\prec 0$.
Moreover, $\Upsilon\prec 0$ implies $\Delta V < 0$. 
It is immediate that the LMI \eqref{eqthm2} guarantees the stability of \eqref{s} for all $[A\; B]\in \Sigma_i$.
It follows from the same characteristics of \eqref{delta} and \eqref{s} that the stability of the closed-loop system \eqref{delta} is ensured.
Therefore, it can be drawn that if there exists some scalar $\beta > 0$ such that \eqref{eqthm2} holds, the tracking error $\delta_i(t)\rightarrow 0$ as $t\rightarrow \infty$ for all $[A\; B]\in \Sigma_i$.
Since the set $\Sigma_i$ contains the true system matrices $A_{\rm tr}$, $B_{\rm tr}$, the MAS \eqref{mas} achieves leader-following consensus asymptotically, completing the proof.
\end{proof}

Thm.~\ref{thm2} provides a tractable tool in terms of LMIs to compute a stabilizing controller gain $K$ and a triggering matrix $\Phi$ directly from data.   
Intuitively speaking, if the sufficient condition $\eqref{eqthm2}$ is satisfied, leader-following tracking can be achieved under the proposed data-driven STC \eqref{controller} and \eqref{trigger_data}.

	\begin{remark}[\emph{Model-based consensus and design}]
		\label{rmk:model}
		The asymptotic stability of the model-based system \eqref{delta} under the STM \eqref{trigger} can be guaranteed via the LMI $\Upsilon\prec 0$ by replacing $[A\; B]$ with the true matrices $[A_{tr}\;B_{tr}]$.
	\end{remark}

\begin{remark}[\emph{Discussion on parameter selection}]\label{rmk:parameter}
		In Thm. \ref{thm2}, the optimization of scalar $\beta$ and matrices $P$, $\bar{\Phi}$, $G$, and $K_G$ is accomplished by solving the LMIs in \eqref{eqthm2}, for carefully chosen parameters $\sigma$ and $\epsilon$. The feasibility of \eqref{eqthm2} hinges upon the existence of suitable values for these parameters.
		Consequently, assuming the existence of feasible solutions in Thm. \ref{thm2}, we proceed to design appropriate values for $\sigma$ and $\epsilon$ that ensure the desired system performance while minimizing sampled-data transmissions. To this end, we suggest the following considerations. By reducing $\sigma$, both the transmission frequency and the rate of state convergence increase, mirroring the impact of parameter $\delta$ discussed in \cite{Peng2013}. To achieve a balance between transmission frequency and system performance, one can employ a grid search approach within the LMI solver. Specifically, starting with a small $\sigma$, it can be incrementally increased until the desired trade-off is achieved.
		Moreover, as mentioned in \cite{Wang2022self}, the scalar $\epsilon$ incorporated within the matrix $\mathcal{R}= [I_n, \epsilon I_n, \mathbf{0}_{n\times n}]^\top$ is introduced to enhance the feasibility of Thm. \ref{thm2}. The selection of $\epsilon$ can follow a similar approach to that of $\sigma$. By gradually adjusting its value, a suitable balance can be attained, ensuring both feasibility and desirable system performance.
\end{remark}

\begin{remark}[\emph{Static data-driven STC design}]
\label{rmk:static}
The proposed data-driven STC adopts a static approach for co-designing the controller and the STM.
Specifically, the controller gain $K$ and the triggering matrix $\Phi$ are designed offline from pre-collected data $\{x_i(T)\}_{T=0}^{\rho}$, $\{u_i(T)\}_{T=0}^{\rho-1}$ of agent $i$, and subsequently implemented in real time to achieve asymptotic consensus of the unknown MAS \eqref{mas}.
Considering the homogeneity of the considered MAS, under the same upper bound on noise, the set of allowable systems corresponding to the obtained stability conditions includes all agents in the network. 
Thus, the asymptotic consensus can be achieved by designing the controller using the data from any agent.
\end{remark}

\begin{remark}[\emph{Direct/indirect approach}]
\label{indirect}
When lacking knowledge of system dynamics, our data-driven STC, as a \emph{direct approach}, offers a new perspective to system analysis and design directly from data. 
Yet, a competing alternative to solve this conundrum is the so-called \emph{indirect approach}, consisting of system identification and model-based STC. 
Although the indirect approach is modular and well-understood, system modeling and identification could be inaccurate and time-consuming, especially when dealing with large-scale systems with scarce, noisy data \cite{Baggio2021}.
\end{remark}

\subsection{Comparison with the model-based STC}
\label{section3d}
According to the data-driven setup in Sec.~\ref{section2c}, we gather state-input data from the \emph{perturbed} open-loop system \eqref{mas_noise} locally in order to simulate the noise, which is inherent to all real-world experiments when recording data.
Our subsequent analysis in Thm.~\ref{thm2} focuses on the stability of the \emph{unperturbed} closed-loop system \eqref{delta}.
In this part, for the sake of comparing our proposed data-driven STC with the model-based one  fairly, we consider the MAS \eqref{mas} with noise
\begin{equation}
\label{mas_perturb}
\begin{split}
	\left\{\begin{aligned}
	{x}_{i}(t+1)\!&=\!A_{\rm tr} x_{i}(t)+B_{\rm tr} u_{i}(t)+B_d d_i(t), \, i=1,2,\ldots,N \\ 
	{x}_{0}(t+1)\!&=\!A_{\rm tr} x_{0}(t)\end{aligned}\right.
	\end{split}
\end{equation}
where  $B_d$ is a known matrix, and $d_i(t)\in \mathcal{R}^n$ is the external disturbance obeying $d_i(t)\in \mathcal{L}_2 [0,\infty ]$. Here, the communication topology of the MAS \eqref{mas_perturb} adheres to Assumption~\ref{graph}. 


Recalling the definition of $e_i(s_k^i)$ of follower $i$,  one has for $t\in\mathbb{N}_{[t_k^i,\,t_{k+1}^i-1]}$ that
\begin{equation}
\label{e1}
	e_{i}(s_k^i)=\eta_1(s_k^i,\,x_i(t_k^i))+\eta_2(s_k^i,\,z_i(t_k^i))+\eta_3(s_k^i,\,d_i(t))
\end{equation}
where 
\begin{align*}
\eta_1(s_k^i,\,x_i(t_k^i))&=(A_{\rm tr}^{s_k^i}-I)x_i(t_k^i), \\
\eta_2(s_k^i,\,z_i(t_k^i))&=\underline{B}_{\rm tr}^{s_k^i}\underline{K}^{s_k^i}z_i(t_k^i),\\
\eta_3(s_k^i,\,d_i(t))&=\sum_{\tau=t_k^i}^{t_k^i+s_k^i-1}A_{\rm tr}^{t-\tau-1}B_d d_i(\tau).
\end{align*}

Here, we remark that at the current triggering time $t_k^i$, we do not know  $e_i(s_k^i)$ because the disturbance $d_i(t)$ is unknown. 
The following lemma presents a method to estimate the upper bound of $e_i(s_k^i)$. For notational brevity, we use $\eta_1$, $\eta_2$, and $\eta_3$ to represent $\eta_1(s_k^i,\,x_i(t_k^i))$, $\eta_2(s_k^i,\,x_i(t_k^i))$, and $\eta_3(s_k^i,\,x_i(t_k^i))$, respectively.

\begin{lemma}[Model-based triggering condition with noise]
\label{upper}
Consider the leader-following MAS \eqref{mas_perturb} and the model-based STC scheme \eqref{controller}, \eqref{trigger} under graph $\bar{\mathcal{G}}$. Suppose Assumption~\ref{graph} holds. For $t\in\mathbb{N}_{[t_k^i,\,t_{k+1}^i-1]}$, if $\sup\|d_i(t)\|\leq \bar{d}_i(t)$ for follower $i$, the following condition holds
\begin{equation*}
	e_{i}^\top(s_k^i) \Phi e_{i}(s_k^i)\leq \eta_0(s_k^i,\,x_i(t_k^i),\bar{d}_i(t))
\end{equation*}
where $
	\Xi =\sum_{\tau=t_k^i}^{t_k^i+s_k^i-1}\Phi^{\frac{1}{2}}A_{\rm tr}^{t-\tau-1}B_d$, and $
	\eta_0 =2(\eta_1+\eta_2)^\top  \Phi(\eta_1+\eta_2)+2\bar{d}^\top _i\Xi^\top \Xi \bar{d}_i$.
\end{lemma}

\begin{proof}
It follows from \eqref{e1} that
\begin{equation}
\label{e2}
\begin{split}
	e_{i}^\top(s_k^i) \Phi e_{i}(s_k^i)&=(\eta_1+\eta_2+\eta_3)^\top \Phi(\eta_1+\eta_2+\eta_3)\\ 
	&=(\eta_1+\eta_2)^\top \Phi(\eta_1+\eta_2)\\
	&\quad+2(\eta_1+\eta_2)^\top  \Phi \eta_3	+\eta_3^\top \Phi \eta_3.
	\end{split}
\end{equation}

The second term in \eqref{e2} can be bounded as follows
\begin{equation}
\label{e3}
	2(\eta_1+\eta_2)^\top  \Phi \eta_3 \leq (\eta_1+\eta_2)^\top  \Phi(\eta_1+\eta_2)+\eta_3 ^\top \Phi\eta_3.
\end{equation}

Moreover, the third term in \eqref{e2} obeys
\begin{align*}
	\eta_3^\top \Phi \eta_3=&\sum_{\tau=t_k^i}^{t_k^i+s_k^i-1} d^\top _i(\tau){A_{\rm tr}^{t-\tau-1}B_d}^\top \Phi \\
	&\quad \times \sum_{\tau=t_k^i}^{t_k^i+s_k^i-1} A_{\rm tr}^{t-\tau-1}B_d d_i(\tau)\leq \bar{d}^\top _i(t) \Xi^\top \Xi \bar{d}_i(t)
\end{align*}
where $\bar{d}_i(t)$ is an upper bound on $d_i(t)$.
Substituting the above inequality and \eqref{e3} into \eqref{e2} yields 
\begin{equation*}
	e_{i}^\top(s_k^i) \Phi e_{i}(s_k^i)\leq 2(\eta_1+\eta_2)^\top  \Phi(\eta_1+\eta_2)+2\bar{d}^\top _i\Xi^\top \Xi \bar{d}_i,
\end{equation*}
which completes the proof.
\end{proof}


By means of Lem.~\ref{upper}, each follower $i$ computes $s_k^i$ at $t_k^i$ by solving
\begin{equation*}
	\eta_0(s_k^i,\,x_i(t_k^i),\bar{d}_i(t))=\sigma z_i^\top(t_k^i) \Phi z_{i}(t_k^i)
\end{equation*}
which gives rise to the next triggering time $t_{k+1}^i=s_k^i+t_k^i$. In other words, our model-based STM for the perturbed MAS \eqref{mas_perturb} can be restated as follows 
\begin{equation}
\label{trigger3}
t_{k+1}^i=t_k^i+\inf_{s_k^i\in \mathbb{N}} \left\{s_k^i\geq 1\big|\tilde{f}_i\left(x_i(t_k^i),\,z_i(t_k^i),\, s_k^i,\bar{d}_i(t)\right)\geq0 \right\}
\end{equation}
with the following triggering function
\begin{align*}
	\tilde{f}_i(x_i(t_k^i),\,z_i(t_k^i),\,s_k^i,\bar{d}_i(t))	&=\eta_0(s_k^i,\,x_i(t_k^i),\bar{d}_i(t))\nonumber
	\\&\quad -\sigma z_i^\top(t_k^i) \Phi z_{i}(t_k^i).
\end{align*}

After adjusting the model-based STM, we turn our attention to the controller design and stability analysis of the perturbed MAS \eqref{mas_perturb}. 
First, associated with the feedback controller \eqref{controller}, the closed-loop tracking error system of agent $i$ is given by
\begin{equation*}
\delta_i(t+1)=A_{\rm tr}\delta_i(t)+B_{\rm tr}Kz_i(t_k^i)+B_d d_i(t),~t\in\mathbb{N}_{[t_k^i,\,t_{k+1}^i-1]}.
\end{equation*}

Consider the definition of distributed $\mathcal{H}_{\infty}$-consensus control for MASs in \cite[Def. 2]{Elahi2019}. 
Leader-following consensus of MAS \eqref{mas_perturb} with external disturbance $d_i(t)$ is transformed into the robust $\mathcal{H}_{\infty}$-control of the tracking error system.
Under the model-based STC scheme \eqref{controller} and \eqref{trigger3}, the following model-based $\mathcal{H}_{\infty}$ stability condition comes ready whose proof is similar to that of the model-based approach in \cite{Elahi2019} and is thus omitted here.


\begin{thm}[Model-based $\mathcal{H}_{\infty}$-consensus and design]
\label{thm3}
Consider the leader-following MAS \eqref{mas_perturb} under graph $\bar{\mathcal{G}}$. Suppose Assumption~\ref{graph} holds. For a given scalar $\sigma>0$ and a disturbance attenuation $\gamma>0$, the $\mathcal{H}_{\infty}$-consensus is achieved for any initial state with the feedback controller \eqref{controller} and the model-based STM \eqref{trigger3}, if there exist positive definite matrices as in Thm.~\ref{thm2} such that the following LMI is satisfied
\begin{equation*}
	\left[\begin{matrix} \widetilde{\Upsilon}+(I_N\otimes L_1^\top L_1)  &I_N\otimes\mathcal{R}B_dG \\ *&-\gamma^2I \end{matrix}\right]\prec0
\end{equation*}
where $\widetilde{\Upsilon}:=\Omega + \Psi + {\rm Sym}\big\{(I_N\otimes \mathcal{R}A_{\rm tr}GL_1)+(\mathcal{H}\otimes \mathcal{R}B_{\rm tr}K_GL_3)\big\}.$ Thus, $K=K_GG^{-1}$ is the  controller gain.
\end{thm}
	
	\begin{remark}[\emph{Distributed control}]\label{rmk:distributed}
Both the controller \eqref{controller} and the data-driven STM \eqref{trigger_data} (model-based STM \eqref{trigger3}) are implemented online in a distributed fashion, where only local (neighbor-to-neighbor) communications are performed.
	Nonetheless, our design and stability analysis (c.f. Thms. \ref{thm2} and \ref{thm3}) rely on certain global information of the MAS, in terms of the Laplacian matrix associated with the communication graph.
	In this sense, our control paradigm involves centralized design and distributed execution.
	It is reasonable when considering a scenario, in which minimal global information is available offline and local data can be acquired during the system operation online, which has been widely studied in model-based STC works \cite{Yi2019,Cui2023,Zhang2022self,Zegers2022}.
	Exploring methods to eliminate the dependence on such global information in data-driven STC design is an area of focus for future research.
\end{remark}

	\begin{remark}[\emph{Consideration of offline/online data noise}]
	The presence of noisy data is explicitly taken into account during the offline data acquisition phase, but not during the online implementation of the designed controller. Similar settings have been popularly adopted in existing research on data-driven control, as evidenced by works in \cite{Waarde2022,Wildhagen2022,Wang2021,Wang2022self,Wang2023jas}.
	However, it is worth noting that it is straightforward to extend the analysis to encompass both online and offline noise sources, thereby providing data-driven robustness and $\mathcal{H}_{\infty}$-performance guarantees. This can be achieved by combining the model-based $\mathcal{H}_{\infty}$ stability condition  (Thm. \ref{thm3}) with the data-driven representation of MASs (Lem. \ref{generallem}).
\end{remark}

\begin{remark}[\emph{Comparison with prior art}]
	Several existing results have explored data-driven STC methods without utilizing system models, such as the trajectory prediction approach \cite{Wenjie2023} and the behavioral approach \cite{Wang2022self} for linear systems.
	In comparison with these works, our approach exhibits several key distinctions. 
	Firstly, our data-driven STC method is robust to noisy data, eliminating the need for noise-free data as required in \cite{Wenjie2023}. 
	Secondly, as mentioned in Remark \ref{rmk:static}, we propose a static co-design method utilizing historical data while ensuring system stability, thereby circumventing the need for real-time iterative computations using online data and reducing the computational burden.
	Finally, to the best of our knowledge, distributed data-driven STC using behavior theory has not been  investigated.
	Unlike the focus in \cite{Wenjie2023, Wang2022self} on single systems, our data-driven approach deals with the asymptotic consensus problem of MASs.
\end{remark}

\section{Simulation Results}
\label{section4}

In this section, we validate the data-driven results and conduct a comparative analysis with the system identification-based and model-based approaches. The numerical computations are carried out in MATLAB, utilizing the YALMIP optimization toolbox along with the SeDuMi solver \cite{Jos1999}.


Consider a MAS composed of six follower pendulums indexed by $1,\,2,\,\ldots,\,6$ and a leader pendulum indexed by $0$; see Fig.~\ref{fig1}.
The dynamics of each pendulum can be approximated by the following linearized equation \cite{Liu2021}
\begin{equation}
\label{simu1}
	\left\{\begin{aligned}
	m\ell^2\ddot{\alpha}_i&=-mg\ell\alpha_i -u_{i}, \quad i=1,\ldots,6, \\ 
	m\ell^2\ddot{\alpha}_0&=-mg\ell\alpha_0,\end{aligned}\right.
\end{equation}
where $g=9.8{\rm{m}}/{\rm{s}}^2$ denotes the gravitational acceleration constant, $m$ and $l$ are the mass and length of each pendulum, $\alpha_i$ and $\alpha_0$  are the pendulum angles of follower $i$ and the leader, and $u_i$ is the control torque of follower $i$.
Referring to \cite{Liu2021}, we select $\ell = 1$m and $m = 1$kg. 
Let $x_i=[{\alpha}_i~\dot{\alpha}_i]^\top$. 
The continuous linearized pendulum model in a periodic sampled-data setting can be described by the discrete-time MAS  \eqref{mas} with
$$A_{\rm tr}=\left[\begin{matrix}0.9980&0.02\\-0.1959&0.9980 \end{matrix}\right],\quad
B_{\rm tr}=\left[\begin{matrix}0.0002\\-0.02\end{matrix} \right]$$
with the sampling period being $T_k=0.02\rm{s}$.
Their communication topology is depicted in Fig.~\ref{fig2}, and satisfies Assumption~\ref{graph}. 

\begin{figure}[!htb]
\centering
\includegraphics[width=2.1in]{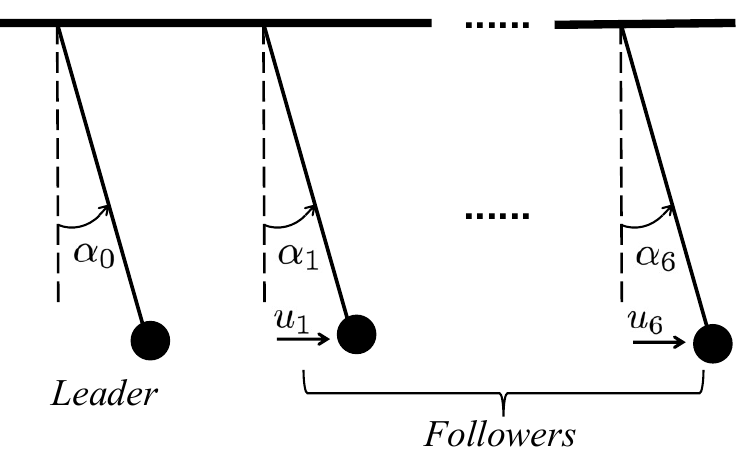}
\caption{The multi-pendulum system consisting of six followers and a leader.}
\label{fig1}
\end{figure}

\begin{figure}[!htb]
\centering
\includegraphics[width=2.1in]{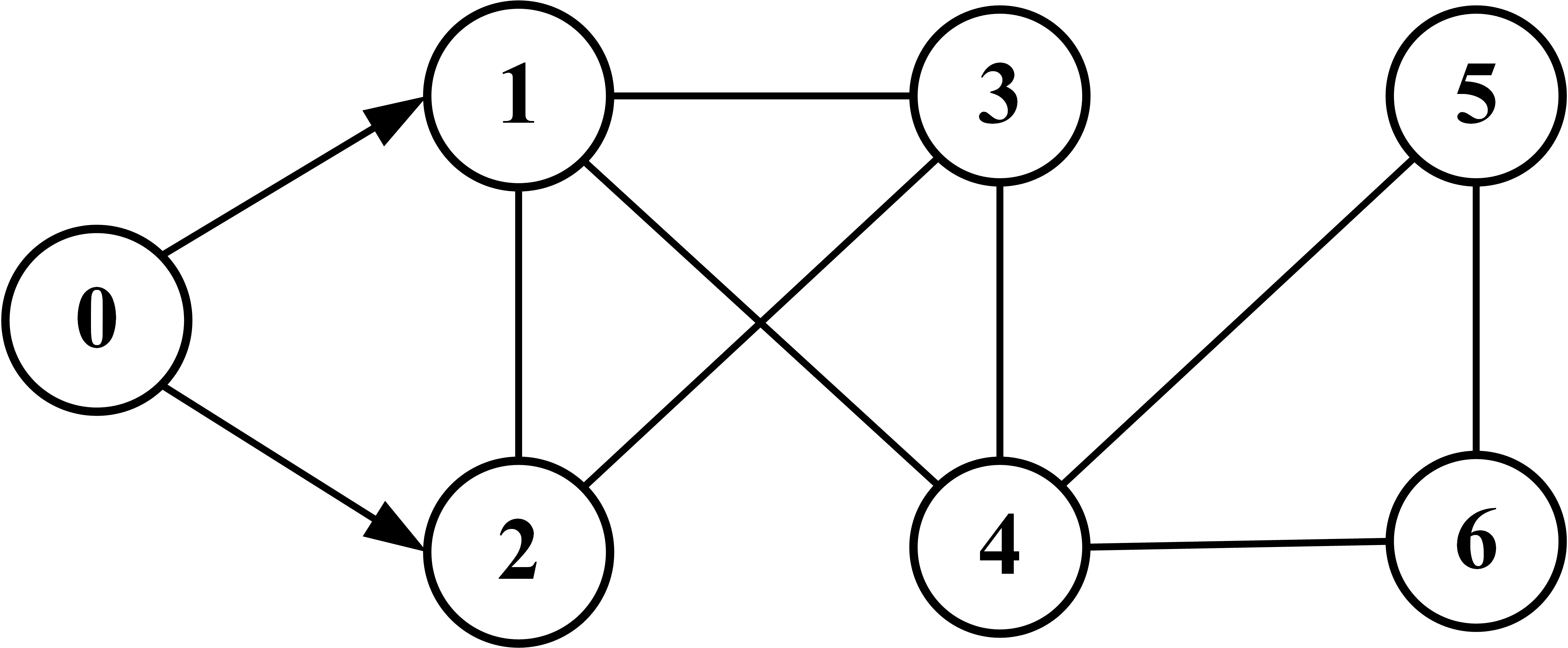}
\caption{The communication topology $\bar{\mathcal{G}}$ between pendulums.}
\label{fig2}
\end{figure}

\begin{figure}[!htb]
	\centering
	\subfloat[]{\includegraphics[scale=0.4]{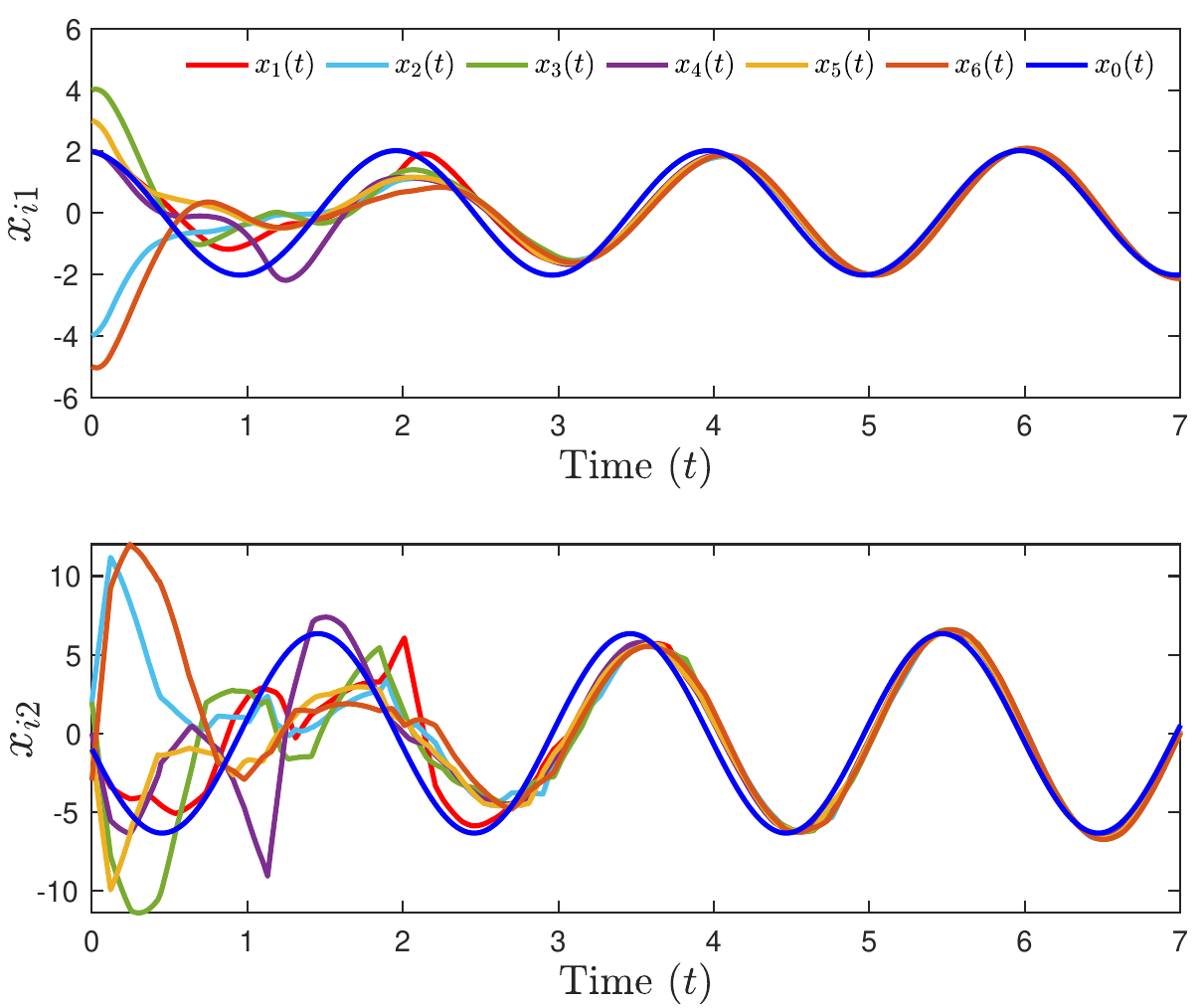}
		\label{fig3a}}\vspace{-4.1mm}
	\hfil
	\subfloat[]{\includegraphics[scale=0.4]{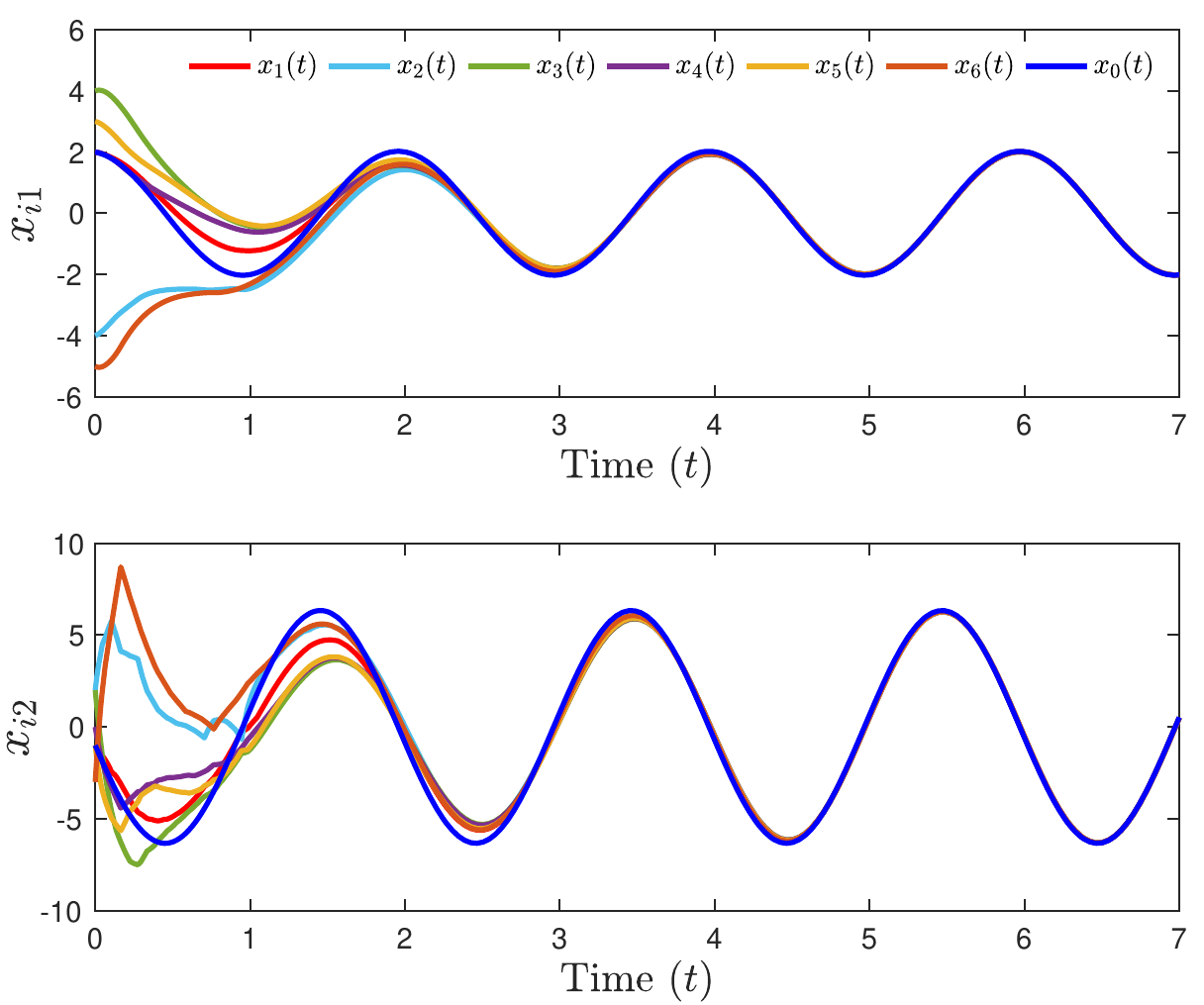}%
		\label{fig3b}}\vspace{-4.1mm}
	\hfil
	\subfloat[]{\includegraphics[scale=0.4]{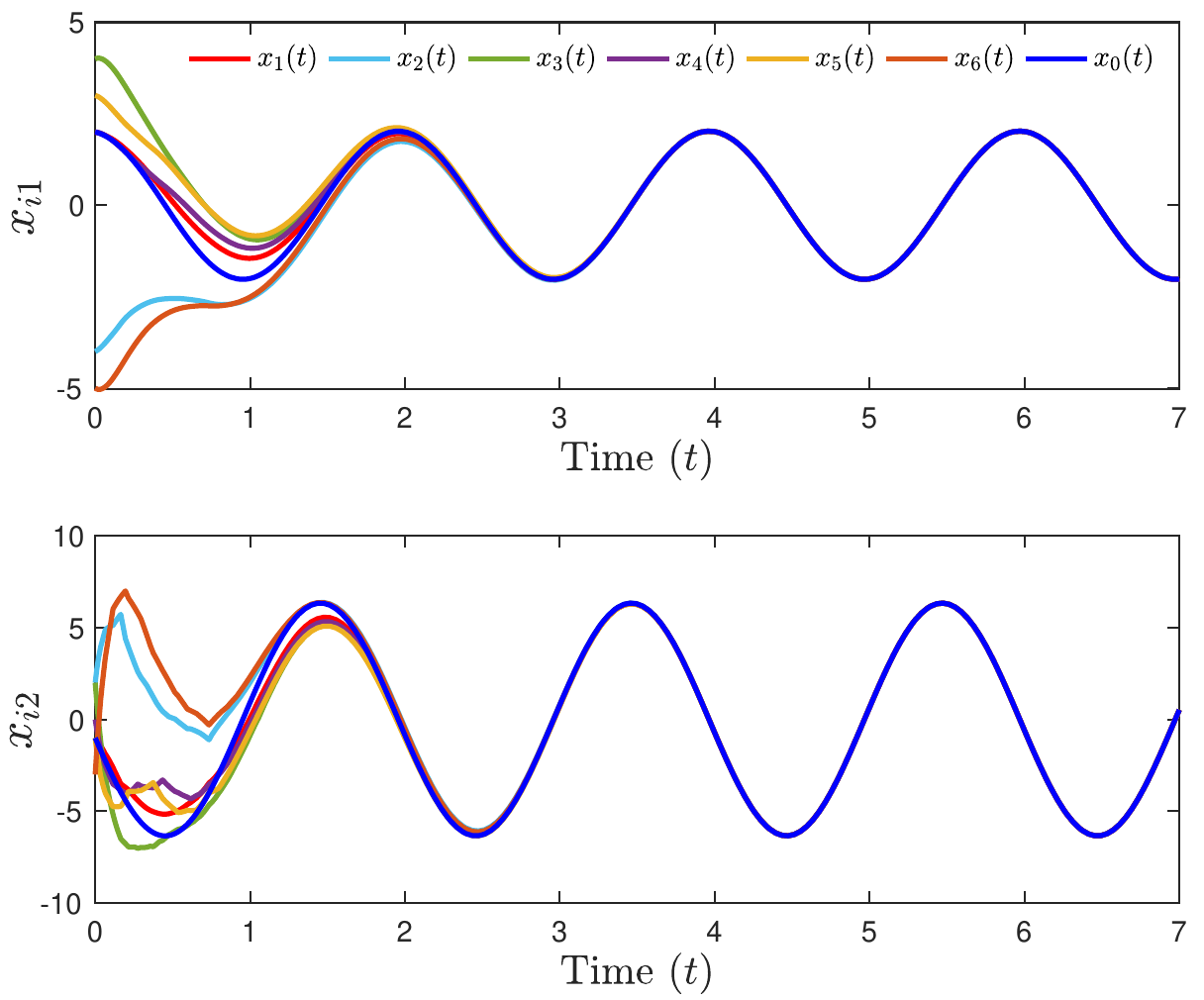}%
		\label{fig3c}}
	\caption{{The state trajectories of each pendulum under the data-driven STC: (a) $\rho=10$; (b) $\rho=80$; (c) $\rho=800$.}}
	\label{fig3}
\end{figure}
\subsection{Consensus under the data-driven STC}
\label{simudata}
The multi-pendulum system was first simulated by using the data-driven STC \eqref{controller} and \eqref{trigger_data}. 
The true system matrices $A_{\rm tr},B_{\rm tr}$ are assumed unknown in the data-driven approach. 
The state-input measurements $\{x_i(T)\}_{T=0}^{\rho}$, $\{u_i(T)\}_{T=0}^{\rho-1}$ for each pendulum are collected from the perturbed system \eqref{mas_noise}, where the data-generating input was generated uniformly from $u_i(t)\in [-1,\, 1]$. 
In order to evaluate the effect of data length on the control performance, varying data lengths including $\rho = 10$, $\rho = 80$ and $\rho= 800$ were employed. 
The noise samples $w_i(t)$ are bounded by $\|w_i(t)\|\leq \bar{w}(t)=0.01$ for all $t\in \mathbb{N}$ and $E=0.01I$, which can be described as a QMI with $Q_d=-I$, $S_d=0$, and $R_d=\bar{w}^2\rho I$ in Assumption \ref{noiseass}.
The parameters were selected as $\epsilon=2$ and $\sigma=0.2$. 
By solving the data-based LMI in Thm. \ref{thm2} for different data lengths $\rho$, the associated controller gain $K$ in \eqref{controller} and the triggering matrix $\Phi$ in \eqref{trigger_data} were found as follows
\begin{align*}
	\rho &=10, {K=[-12.1805\, -9.4375],\Phi=\left[\begin{matrix}
	77.9309&42.4900\\
	42.4900&83.9906
\end{matrix}\right]}\\
\rho &=80, {K=[-2.9531\, -3.1986],\Phi=\left[\begin{matrix}
	39.4607&13.6684\\
	13.6684&35.1941
\end{matrix}\right]}\\
\rho &=800, {K\!=\![-1.2594\, -0.7333], \Phi\!=\!\left[\begin{matrix}
	23.7346&\!6.4973\\
	6.4973&\!23.8443
\end{matrix}\right].}
\end{align*}
In addition, setting $\bar{s}=40$, we adopted \cite[Alg. 1]{Wildhagen2022}  to compute $\Theta_i^s$ for each follower and for all $s\in\mathbb{N}_{[1,\bar{s}]}$, and then obtained the data-driven STM.

Let the initial states of the pendulums be $x_0(0)=[2,\,-1]^\top$, $x_1(0)=[-4,\,2]^\top$, $x_2(0)=[4,\,2]^\top$, $x_3(0)=[2,\,0]^\top$, $x_4(0)=[3,\,-1]^\top$, $x_5(0)=[-5,\, -3]^\top$, and $x_6(0)=[2,\, 0.5]^\top$, respectively.
The state trajectories of each pendulum under the data-driven control with different data lengths are presented in {Fig.~\ref{fig3}}.
Obviously, leader-following consensus is achieved in all three settings, which validates the correctness of the data-driven STC. 
Besides, it can be observed that the larger the datasize, the better the control performance, in the sense that having a faster convergence speed, smoother trajectories, and smaller convergence errors.
This is because a larger dataset provides more information about the system dynamics, thereby reducing the uncertainties caused by noise. As a result, the set of allowable system matrices in $\Sigma_i$ becomes smaller as $\rho$ increases, which mirrors the results in \cite{Wildhagen2022}.




\subsection{Comparing with the identification-based STC}
 We tested the indirect approach, namely the identification-based STC,  described in Remark~\ref{indirect} on the multi-pendulum system.
The indirect approach comprises a system identification step using e.g., the subspace space system identification (n4sid) technique, followed by the model-based STC design.
Specifically, we first estimated a discrete-time state-space model for the multi-pendulum system using the n4sid toolbox in MATLAB based on the same set of data in Sec.~\ref{simudata}.
Theoretical analysis of the n4sid identification algorithm can be found in \cite{VanOverschee1996}.
The controller gain $K$ in \eqref{controller} and the triggering matrix $\Phi$ in \eqref{trigger} were designed according to Thm.~\ref{thm3} for $\rho=10$, $\rho=80$, and $\rho=800$ as follows
\begin{align*}
	\rho &=10, K=[-14.3887~ -8.9605],\Phi=\left[\begin{matrix}
	29.6795&3.3634\\
	3.3634&1.6010
\end{matrix}\right]\\
\rho &=80, K=[-1.5706~ -1.6844],\;\,\Phi=\left[\begin{matrix}
	1.5081&0.0622\\
	0.0622&1.0853
\end{matrix}\right]\\
\rho &=800, K=[-1.5368~ -0.8696], \Phi=\left[\begin{matrix}
	1.9140&1.0371\\
	1.0373&1.6703
\end{matrix}\right].
\end{align*}

We used the same parameters and initial state for each pendulum as in Sec.\ref{simudata}. Fig.~\ref{fig4} plots the state trajectories of each pendulum. Clearly, leader-following consensus is achieved using the identification-based STC. 
Further, we reported the steady-state time for data lengths $\rho=10$, $\rho=80$, and $\rho=800$ under both the data-driven STC and the identification-based STC in Table~\ref{table1}. 
It can be seen that the steady-state time of  both approaches decreases as the data length increases. 
In particular, for small data length e.g., $\rho=10$, the identification-based approach converges to its steady-state slightly faster than the data-driven one by relying on an explicit system model, which is analogous to the observation made  in \cite{Krishnan2021}.
On the other hand, the data-driven STC with $\rho=80$ has the same steady-state time ($3$s) as the identification-based one using $\rho=800$. 
That is, as soon as there are sufficient data,
the data-driven STC outperforms the identification-based STC.
One can also see from Table~\ref{table1} that there is little change in steady-state time (from $t=3$s to $t=2$s)  for the data-driven STC when  $\rho$ increases from $80$ to $800$, demonstrating the robustness of the data-driven approach.

\begin{figure}[!htb]
	\centering
	\subfloat[]{\includegraphics[scale=0.4]{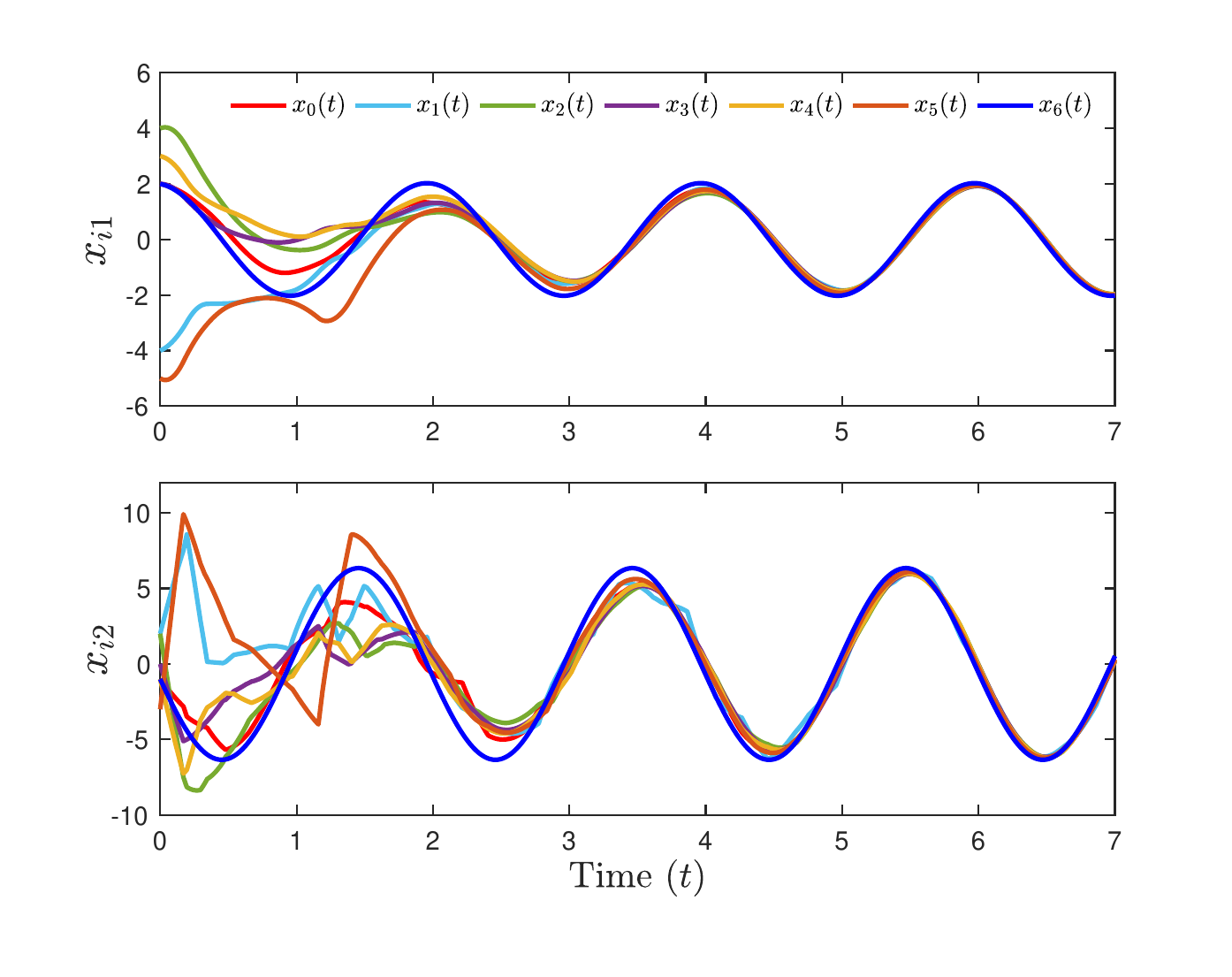}
		\label{fig4a}}\vspace{-4.1mm}
	\\
	\subfloat[]{\includegraphics[scale=0.4]{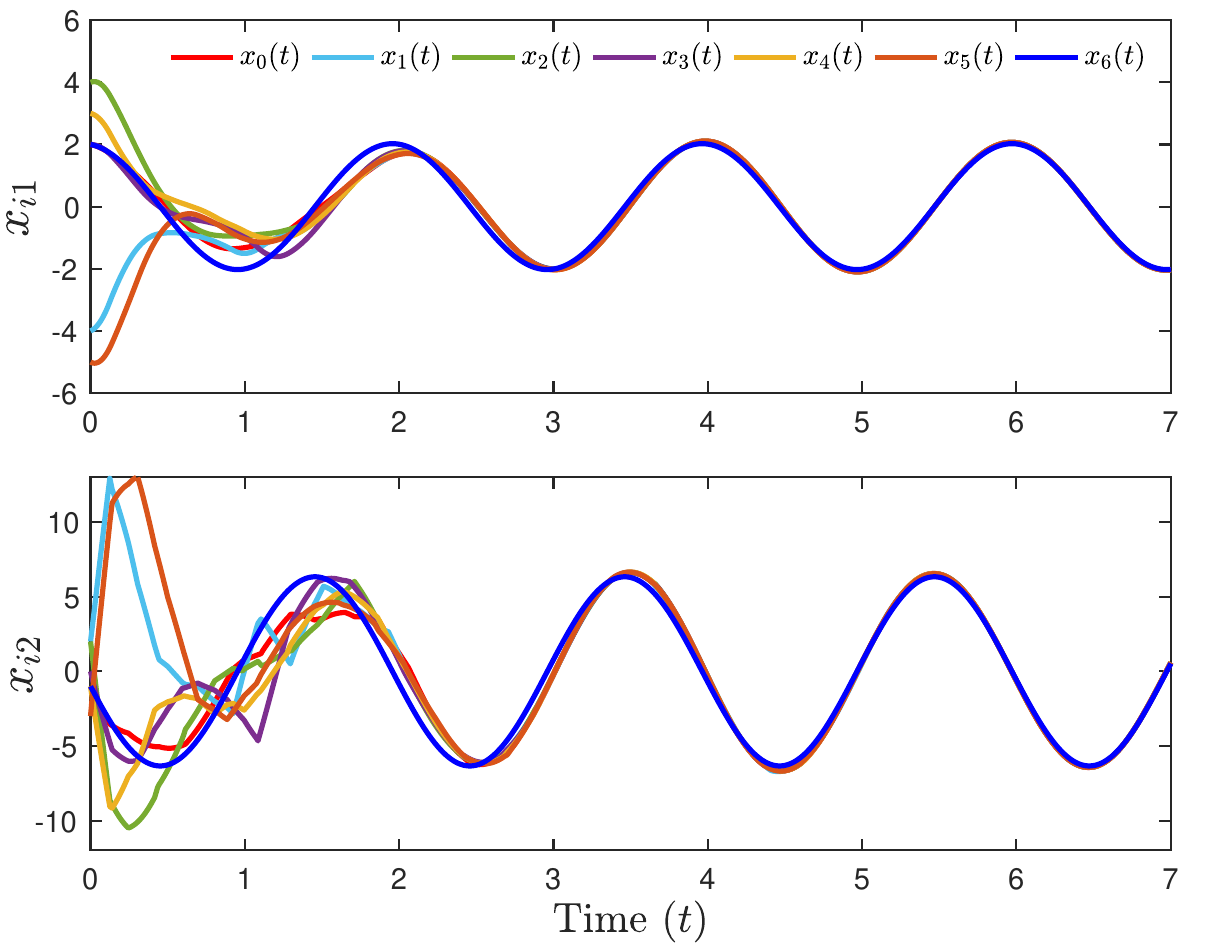}%
		\label{fig4b}}\vspace{-4.1mm}
	\\
	\subfloat[]{\includegraphics[scale=0.4]{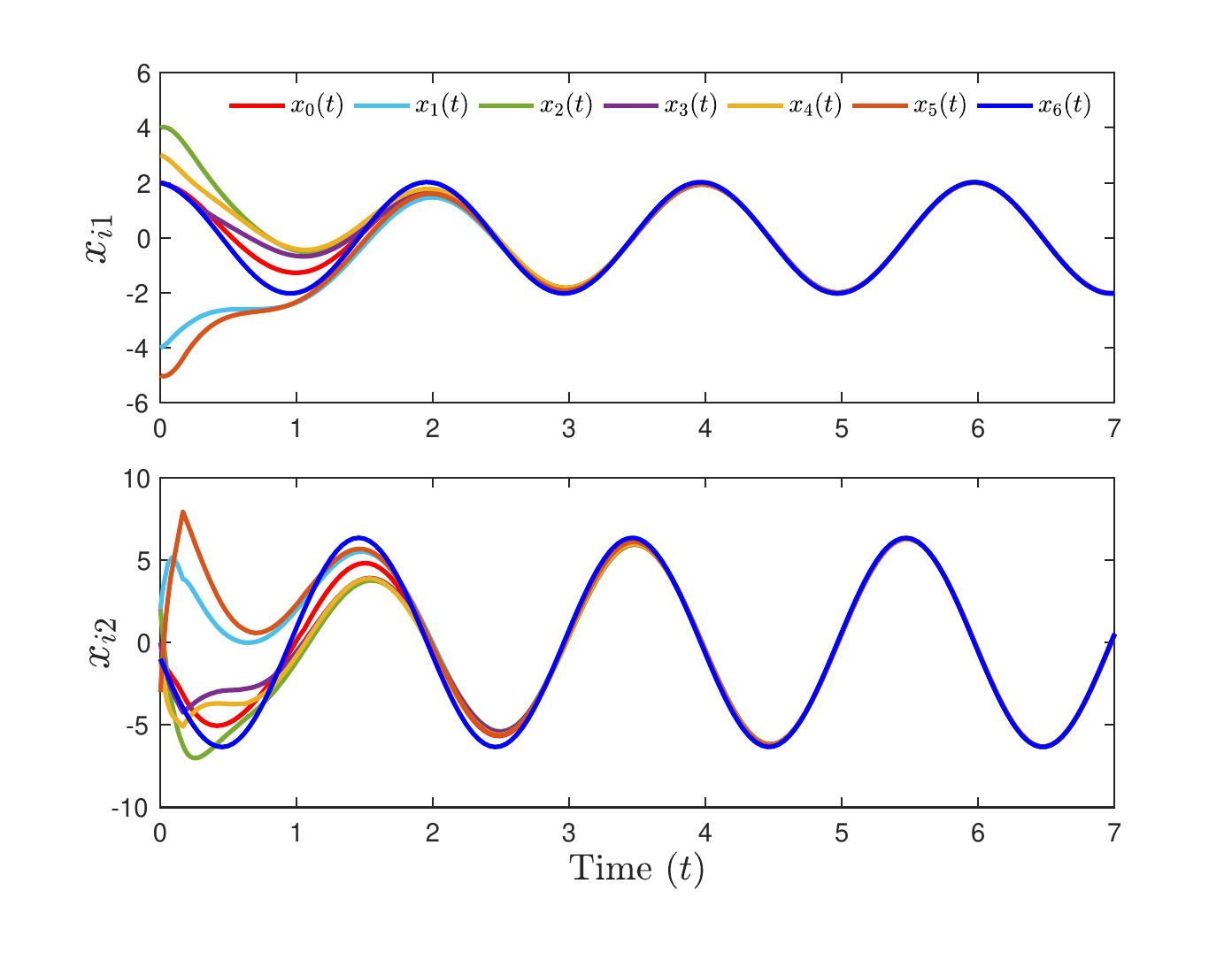}%
		\label{fig4c}}
	\caption{The state trajectories of each pendulum under the identification-based STC: (a) $\rho=10$; (b) $\rho=80$; (c) $\rho=800$.}
	\label{fig4}
\end{figure}

\begin{table}[!htb]
	\centering
	\caption{Steady-state time for different data lengths $\rho$ under data-driven STC and identification-based STC}
	\label{table1}
	\resizebox{\linewidth}{!}{
		\begin{tabular}{ccccc}
			\toprule 
			\multicolumn{2}{c}{The pre-collected state-input data length $\rho$}                    & 10 & 80 & 800                   \\ \midrule 
			\multirow{2}{*}{Steady-state time} & Data-driven STC    & 7  & 3  & \multicolumn{1}{c}{2} \\
			& Identification-based STC  & 6  & 5  & 3   \\
			\bottomrule 
	\end{tabular}	}
\end{table}

\begin{figure}[!htb]
	\centering
	\includegraphics[scale=0.4]{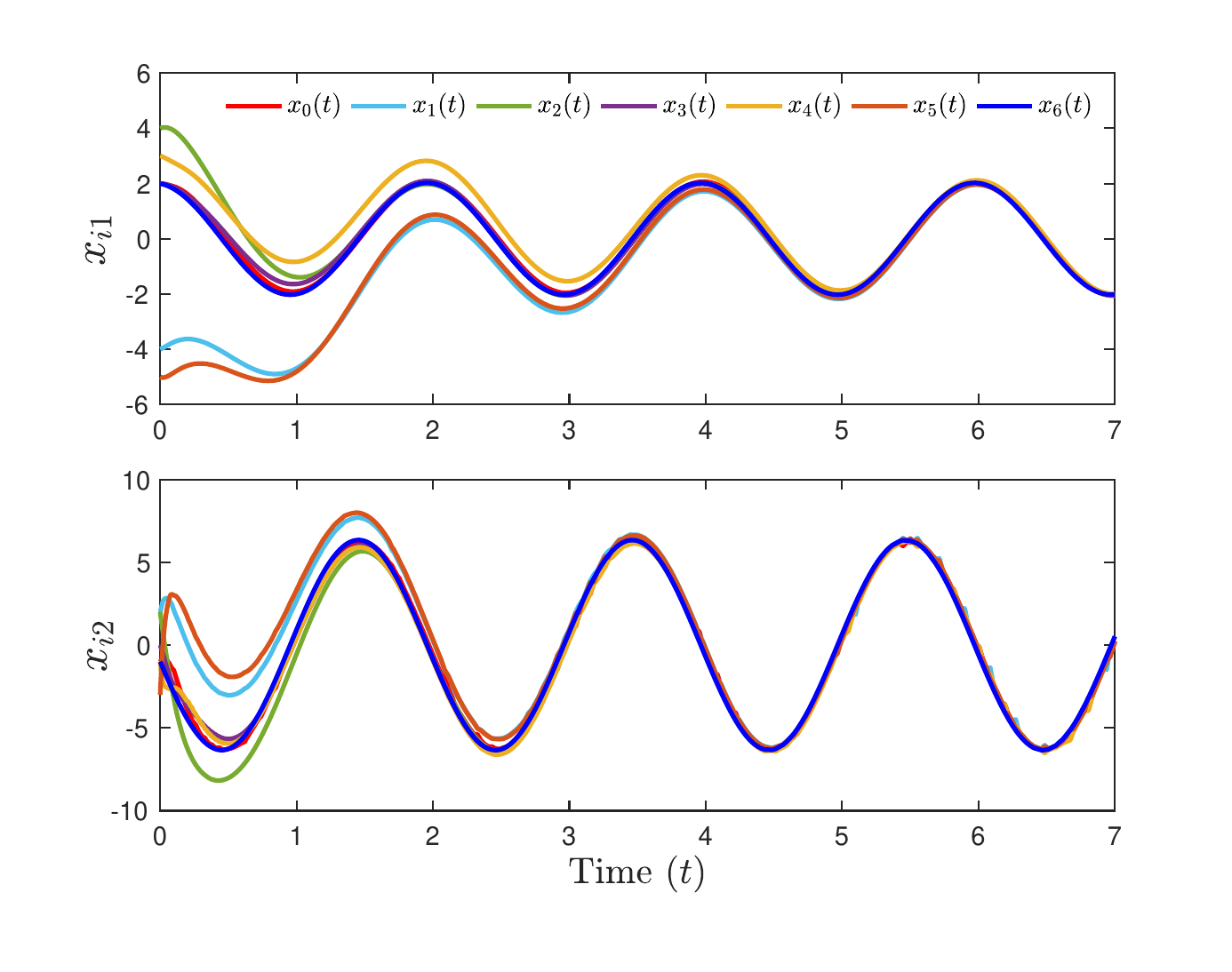}
	\caption{The state trajectories of each pendulum under the model-based STC.}
	\label{fig5}
\end{figure}


\subsection{Comparing with the model-based STC}
In what follows, we consider the multi-pendulum system \eqref{mas_perturb} with bounded disturbance, where the true system matrices $A_{\rm tr},B_{\rm tr}$ are assumed known.
Set $d_i(t)=\bar{d}_i(t)[\sin(3\pi t+\frac{\pi}{8}i),\,\sin(3\pi t+\frac{\pi}{8}i)]^\top$ for $i=1,\,2,\,\ldots,\,6$. 
To ensure a fair and meaningful comparison, select $B_d=E=0.01I$ and the disturbance bound $\bar{d}_i(t)=\bar{w}_i(t)=0.01$.
Taking  $\gamma=1$ and solving the LMI in Thm.~\ref{thm3}, the controller gain $K$ in \eqref{controller} and the triggering matrix $\Phi$ in \eqref{trigger3} were obtained as follows
$$
	K=[-1.0368\, -1.9844],\quad  \Phi=\left[\begin{matrix}8.4320  &  0.3945\\
		0.3945   & 0.8960\end{matrix}\right].
$$
The simulation results are shown in Fig.~\ref{fig5}. 
It is evident that all followers track the leader, validating the effectiveness and robustness of the proposed model-based $\mathcal{H}_{\infty}$-STC.

\subsection{Control performance vs. communication efficiency}
Next, we compare the control performance versus communication efficiency of different STC approaches as well as different triggering parameters.
To this end, let us define the following index to quantitatively assess the control performance
\begin{align*}
J^c(t)=\ln\!\Big(&\sum_{l=0}^{t}\sum_{i=0}^{N}\|x_i(l)\|^2_{Q_i}
+\|u_i(l)\|^2_{R_i}\\&+\|x_i(l)-x_0(l)\|^2_{Q_{i0}}\Big)
\end{align*}
for pre-selected, positive definite weight matrices $Q_i$, $R_i$, and $Q_{i0}$ for $i=1,\,2,\,\ldots,\,N$.
Intuitively, we expect the index to be small, implying that the control goal can be met at a low cost. Taking matrices $R_i=5I$, $Q_i=10I$, and $Q_{i0}=3I$, performance index trajectories $J^c(t)$ associated with the three approaches for different data lengths are presented in Fig.~\ref{fig6}. 
Clearly, as the data length increases, the control performance is enhanced for both the data-driven STC and the identification-based STC. Besides, when $\rho\ge 80$, the performance indices  of the data-driven STC are consistently lower than those of the identification-based one, which corroborates the control effectiveness of the proposed data-driven approach.

\begin{figure}[htb]
	\centering
	\includegraphics[scale=0.37]{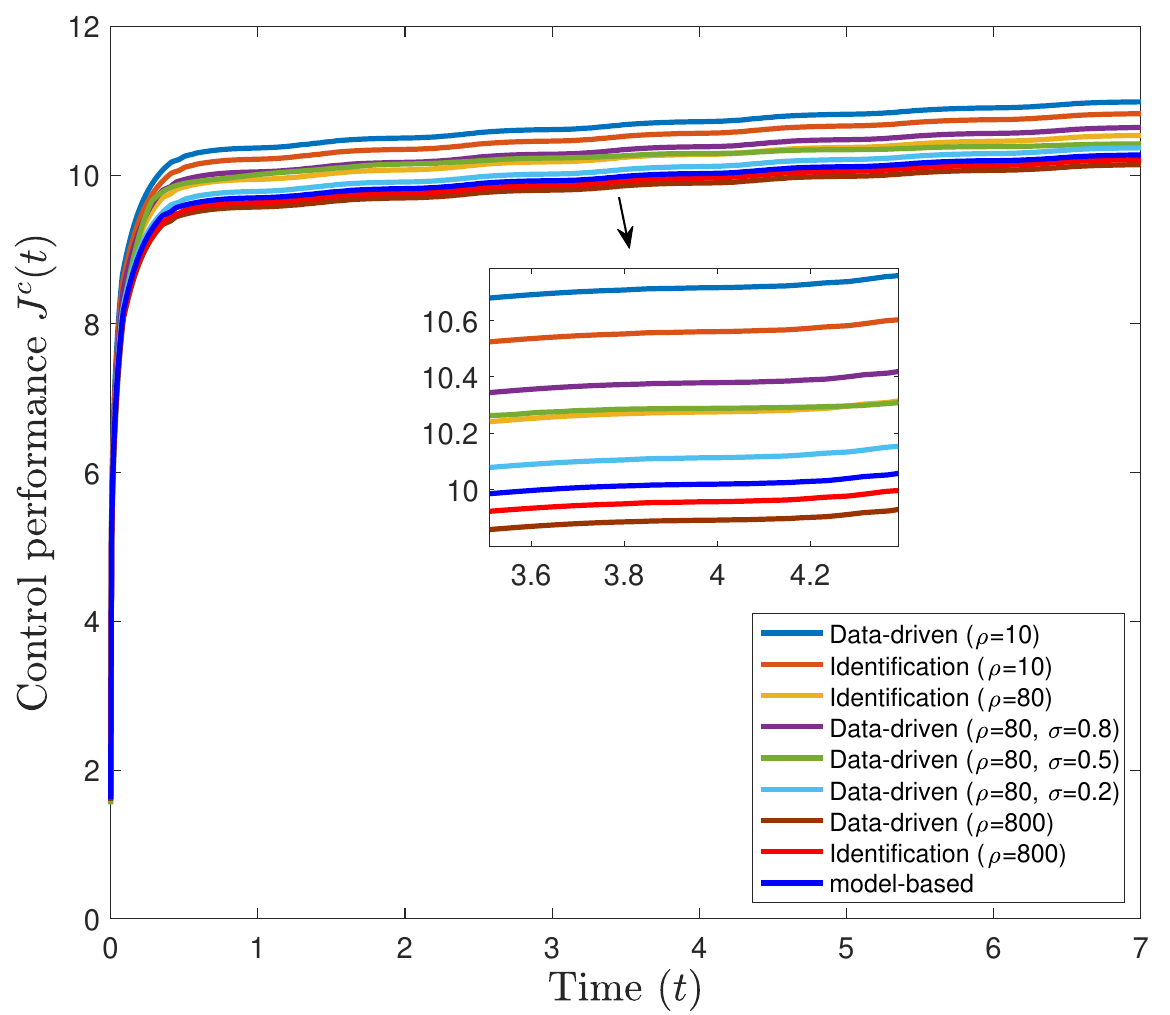}
	\caption{$J_c(t)$ under the data-driven STC, identification-based STC, and model-based STC with (a) $\rho=10$; (b) $\rho=80$; (c) $\rho=800$.}
	\label{fig6}
\end{figure}

\begin{figure}[!htb]
	\centering
	\includegraphics[scale=0.28]{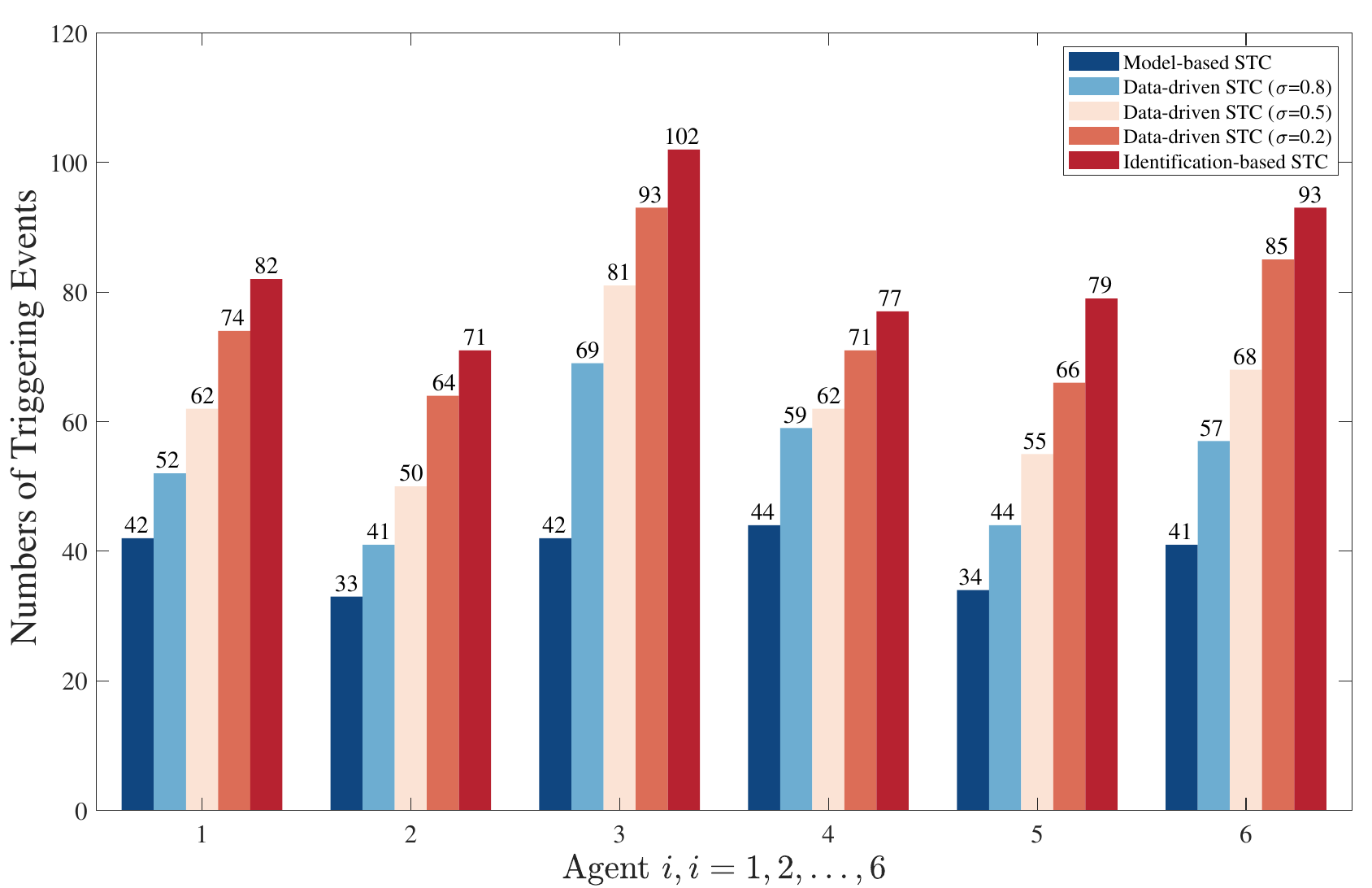}
	\caption{Number of triggering times under the model-based STC, data-driven STC, and system identification-based STC with $\rho=80$ and $\sigma=0.2$.}
	\label{fig7}
\end{figure}

In terms of communication efficiency, we report the numbers of triggering events in Fig.~\ref{fig7} for $\rho=80$ and $\sigma=0.2$. 
It can be observed that the data-driven STC exhibits a slightly higher triggering frequency than the model-based STC, while control performance at the approximate level are guaranteed for both approaches.
The main reason could be that the data-driven representation of lifted MASs incurs a degree of conservatism since it introduces the matrix $M$ and relies on an overapproximation of $\mathcal{W}_i^s$, as stated in Remark \ref{rmk:representation}. That is, the model-based stability condition comes with less conservatism than the data-driven one.
Moreover, in contrast to the data-driven approach, the identification-based scheme requires more communications among agents to achieve consensus. This increased communication frequency is mainly due to the additional operations involved in the identification-then-control procedure, which can potentially amplify the sensitivity to round-off errors \cite{Baggio2021}. Furthermore, the finite and noisy data lead to inaccurate identification results. It is also evident that the upper bound $\bar{d}_i(t)$ on the disturbance employed in the STM \eqref{trigger3} for predicting the next triggering time introduces conservatism.
As a result, the data-driven approach demonstrates superior communication efficiency compared to the identification-based one when dealing with finite, noisy data.
	
	We also compare the system performance and the number of triggerings for different values of $\sigma$, providing verification for the discussion in Remark \ref{rmk:parameter}. Combining the findings in Figs.~\ref{fig6} and \ref{fig7}, we conclude that the proposed data-driven STC achieves a favorable balance between control performance and communication efficiency by appropriately selecting the triggering parameter. 

\section{Conclusions}
\label{section5}
The self-triggered consensus control problem of unknown linear MASs has been investigated in this paper. A data-driven approach was proposed to design the self-triggering transmission policies and controllers directly from  pre-collected, noisy data.
A data-based criterion for co-designing the controller gain and triggering matrix was established to ensure closed-loop system stability.
To evaluate the robustness of data-driven STC and model-based STC against disturbance, a model-based condition was derived.
The effectiveness of the proposed data-driven STC was showed through numerical comparison with the system identification-based STC and the model-based $\mathcal{H}_{\infty}$-STC. 
Our results suggest that
the data-driven method exhibits a good trade-off between control performance and communication efficiency, 
while being more implementation-friendly and data-efficient than the system identification-based method,
especially when dealing with finite, noisy data.
Nonetheless, the presented STM builds on the data-driven representation of lifted MASs, which introduces some conservatism. A promising direction for future research is how to obtain a less conservative data-driven representation for MASs.
Moreover, our results address only the data-driven control problem of linear homogeneous leader-following MASs, which can be extended to general systems, including e.g., heterogeneous MASs or nonlinear MASs.



\bibliographystyle{IEEEtran}
\bibliography{yifeibibfile}

\end{document}